\documentclass[11pt,leqno]{article}
\usepackage{amsthm,amsmath, mathtools}
\usepackage[round]{natbib}
\usepackage{amsfonts}
\usepackage{ifthen}
\usepackage{enumitem}
\usepackage{amssymb,dsfont,url,color,booktabs}
 \usepackage{tikz}
 \usepackage{setspace}

\usepackage{epstopdf}
\usepackage{chngcntr}
\usepackage{bbm}
\usepackage{array}
\usepackage[export]{adjustbox}[2011/08/13] 
\usepackage{geometry}
\usepackage{rotating}
\usepackage{todonotes}
\usepackage{xspace}
\usepackage{algorithm}
\usepackage[noend]{algpseudocode}
 \usepackage[english]{babel}
\usepackage{bbm}
\usepackage{float}
\usepackage{caption}
\usepackage{subcaption}
\usepackage{multirow}
\usepackage{booktabs}
\usepackage[font={footnotesize}]{caption}

\usepackage{placeins} 		

\theoremstyle{plain}
\newtheorem{theorem}{Theorem}[section]

\newtheorem{proposition}[theorem]{Proposition}
\newtheorem{corollary}[theorem]{Corollary}
\newtheorem{defin}[theorem]{Definition}
\newtheorem{remark}[theorem]{Remark}

\newtheorem{assumptions}[theorem]{Assumptions}

\newcommand{\notiz}[1]{\relax}

\newcommand{\zitep}[1]{\relax}

\newcommand{\1}{\mathds 1}            

\newcommand{\nn}{\mathds N}

\newcommand{\cc}{\mathds C}

\newcommand{\Price}[1][]{
		\ifthenelse{\equal{#1}{}}{\mathit{Price}}{\Price{}^{#1}}
	} 

\newlength{\wordlength}

\newcommand{\ul}{\underline}
\newcommand{\ol}{\overline}

\newcommand{\EE}{\mathbb{E}}

\newcommand{\dx}{\text{d}x}

\renewcommand{\cite}{\citet}

\numberwithin{equation}{section}
\numberwithin{figure}{section}
\numberwithin{table}{section}

\begin{document}
\title{\textbf{A new approach for American option pricing: The Dynamic Chebyshev method}
}

\bigskip
\author{\textbf{Kathrin Glau$\vphantom{l}^{1,2}$,} \textbf{Mirco Mahlstedt$\vphantom{l}^{2,*}$,} \textbf{Christian P{\"o}tz$\vphantom{l}^{1,2,}$\footnote{The authors thank the KPMG Center of Excellence in Risk Management for their support.}
}
\\\\$\vphantom{l}^{\text{1}}$Queen Mary University of London, UK\\
$\vphantom{l}^{\text{2}}$Technical University of Munich, Germany
}

\maketitle
\begin{abstract}
We introduce a new method to price American options based on Chebyshev interpolation. In each step of a dynamic programming time-stepping we approximate the value function with Chebyshev polynomials. The key advantage of this approach is that it allows to shift the model-dependent computations into an offline phase prior to the time-stepping. In the offline part a family of generalised (conditional) moments is computed by an appropriate numerical technique such as a Monte Carlo, PDE or Fourier transform based method. Thanks to this methodological flexibility the approach applies to a large variety of models. Online, the backward induction is solved on a discrete Chebyshev grid, and no (conditional) expectations need to be computed. For each time step the method delivers a closed form approximation of the price function along with the options' delta and gamma. Moreover, the same family of (conditional) moments yield multiple outputs including the option prices for different strikes, maturities and different payoff profiles. We provide a theoretical error analysis and find conditions that imply explicit error bounds for a variety of stock price models. Numerical experiments confirm the fast convergence of prices and sensitivities. An empirical investigation of accuracy and runtime also shows an efficiency gain compared with the least-square Monte-Carlo method introduced by Longstaff and Schwartz (2001).
\end{abstract}

	


\section{Introduction}
A challenging task for financial institutions is the computation of prices and sensitivities for large portfolios of derivatives such as equity options. Typically, equity options have an early exercise feature and can either be exercised at any time until maturity (American type) or at a set of pre-defined exercise dates (Bermudan type).
In lack of explicit solutions, different numerical methods haven been developed to tackle this problem.
One of the first algorithms to compute American put option prices in the Black-Scholes model has been proposed by \cite{BrenanSchwartz1977}.  In this approach, the related partial differential inequality is solved by a finite difference scheme. A rich literature further developing the PDE approach has accrued since, including methods for jump models (\cite{Levendorskii2004b}, \cite{HilberReichmannSchwabWinter2013}), extensions to two dimensions (\cite{HaentjensHout2015}) and combinations with complexity reduction techniques (\cite{HaasdonkSalomonWohlmuth2013}). Besides PDE based methods a variety of other approaches has been introduced, many of which trace back to the solution of the optimal stopping problem by the dynamic programming principle, see e.g. \cite{peskir2006optimal}.  For Fourier based solution schemes we refer to \cite{LordFangBervoetsOsterlee2008}, \cite{fang2009pricing}. Simulation based approaches are of fundamental importance, the most prominent representative of this group is the Least-squares Monte-Carlo (LSM) approach of \cite{longstaffschwartz}, we refer to \cite{glasserman2003monte} for an overview of different Monte-Carlo methods. Fourier and PDE methods typically are highly efficient, compared to simulation, however, they are less flexible towards changes in the model and particularly in the dimensionality.

In order to reconcile the advantages of the PDE and Fourier approach with the flexibility of Monte Carlo simulation, we propose a new approach. We consider a dynamic programming time-stepping. Let ${X_{t}}$ be the underlying Markov process and the value function $V_t$ is given by,
\begin{align*}
V_{T}(x)&=g(x)\\
V_{t}(x)&=f\left(g(t,x), \mathbb{E}[V_{t+1}(X_{t+1})\vert X_{t}=x]\right).
\end{align*}
with time steps $t<t+1<\ldots<T$ and payoff function $g$. The computational challenge is to compute  $\mathbb{E}[V_{t+1}(X_{t+1})\vert X_{t}=x]$ for \textit{for all time steps $t$ and all states $x$}, where $V_{t+1}$ depends on \textit{all previous time steps}.

In order to tackle this problem, we approximate the value function in each time step by Chebyshev polynomial interpolation. We thus express the value function as a finite sum of Chebyshev polynomials
\begin{align}\label{eq1}
\mathbb{E}[V_{t+1}(X_{t+1})\vert X_{t}=x] \,\approx  \sum c_j^{t+1} \mathbb{E}[T_j(X_{t+1})\vert X_{t}=x]
\end{align}
The choice of Chebyshev polynomials is motivated by the promising properties of Chebyshev interpolation such as
\begin{itemize}
\item The vector of coefficients $(c_j^{t+1})_{j=0,\ldots,N}$ is explicitly given as a linear combination of the values $V_t(x_k)$ at the Chebyshev grid points $x_k$. 
For this, equation \eqref{eq1} needs to be solved at the Chebyshev grid points $x=x_k$ only. 
\item Exponential convergence of the interpolation for analytic functions and polynomial convergence of differential functions depending on the order.
\item The interpolation can be implemented in a numerically stable way.
\end{itemize}

The computation of the continuation value at a single time step coincides with the pricing of a European option. Its interpolation with Chebyshev polynomials is proposed in \cite{GassGlauMahlstedtMair2018}, where the method shows to be highly promising and exponential convergence is established for a large set of models and option types. Moreover, the approximation of the value function with Chebyshev polynomials has already proven to be beneficial for optimal control problems in economics, see \cite{judd1998numerical} and \cite{CaiJudd2013}.

The key advantage of our approach for American option pricing is that it collects all model-dependent computations in the generalized conditional moments $\Gamma_{j,k}=\mathbb{E}[T_j(X_{t+1})\vert X_{t}=x_k]$. If there is no closed-form solution their calculation can be shifted into an offline phase prior to the time-stepping. Depending on the underlying model a suitable numerical technique such as Monte Carlo, PDE and Fourier transform methods can be chosen, which reveals the high flexibility of the approach. Once the generalized conditional moments $\Gamma_{j,k}$ are computed, the backward induction is solved on a discrete Chebyshev grid. Which avoids any computations of conditional expectations during the time-stepping. For each time step the method delivers a closed form approximation of the price function $x\mapsto\sum c_j^{t}T_j(x)$ along with the options' delta and gamma. Since the family of  generalized conditional moments $\Gamma_{j,k}$ are independent of the value function, they can be used to generate multiple outputs including the option prices for different strikes, maturities and different payoff profiles. The structure of the method is also beneficial for the calculation of expected future exposure which is the computational bottleneck in the computation of CVA, as investigated in \cite{GlauPachonPoetz2018}.

The offline-online decomposition separates model and payoff yielding a modular design. We exploit this structure for a thorough error analysis and find conditions that imply explicit error bounds. They reflect the modularity by decomposing into a part stemming from the Chebyshev interpolation, from the time-stepping and from the offline computation. Under smoothness conditions the asymptotic convergence behaviour is deduced.

We perform numerical experiments using the Black-Scholes model, Merton's jump diffusion model and the Constant Elasticity of Variance (CEV) model as a represenative of a local volatility model. For the computation of the generalized conditional moments we thus use different techniques, namely numerical integration based on Fourier transforms and Monte Carlo simulation. Numerical experiments confirm the fast convergence of option prices along with its delta and gamma. A comprehensive comparison with the LSM reveals the potential efficiency gain of the new approach, particularly when several options on the same underlying are priced.

The rest of the article is organized as follows. We introduce the problem setting and the new method in Section 2 and provide the error analysis in Section 3. Section 4 discusses general traits of the implementation and Section 5 presents the numerical experiments. Section 6 concludes the article, followed by an appendix with the proof of the main result.

\section{The Chebyshev method for Dynamic programming problems}
First, we present the Bellman-Wald equation as a specific form of dynamic programming. Second, we provide the necessary notation for the Chebyshev interpolation. Then we are in a position to introduce the new approach and its application to American option pricing.

\subsection{Optimal stopping and Dynamic Programming}
Let $X=(X_t)_{0\leq t\leq T}$ be a Markov process with state space $\mathbb{R}^{d}$ defined on the filtered probability space $(\Omega,\mathcal{F},(\mathcal{F}_{t})_{t\geq 0},\mathbb{P})$. Let $g:[0,T]\times\mathbb{R}^{d}\longrightarrow\mathbb{R}$ be a continuous function with $\EE\left[\sup_{0\leq t\leq T}\left|g(t, X_{t})\right|\right]<\infty.$ Then
\begin{align*}
V(t,x):=\sup_{t\leq\tau\leq T}\EE\left[g(\tau,X_{\tau})\vert X_{t}=x\right] \qquad \text{for all} \quad (t,x)\in [0,T]\times\mathbb{R}^{d}
\end{align*}
over all stopping times $\tau$, see $(2.2.2')$ in \cite{peskir2006optimal}. In discrete time the optimal stopping problems can be solved with dynamic programming.\\ 

Namely, with time stepping $t=t_0<\ldots<t_{n}=T$ the solution of the optimal stopping problem can be calculated via backward induction
\begin{align*}
V_{T}(x)&=g(T,x)\\
V_{t_u}(x)&=\max\left(g(t_u,x), \EE[V_{t_{u+1}}(X_{t_{u+1}})\vert X_{t_u}=x]\right).
\end{align*}
Note that $n$ refers to the number of time steps between $t$ and $T$. For notational convenience, we indicate the value function at each time step with subscript $t_u$ to directly refer to the time step $t_u$. For a detailed overview of optimal control problems in discrete time we refer to \cite{peskir2006optimal}.

\subsection{Chebyshev polynomial interpolation}

The univariate Chebyshev polynomial interpolation as described in detail in \cite{Trefethen2013} has a tensor based extension to the multivariate case, see e.g. \cite{SauterSchwab2010}. Usually, the Chebyshev interpolation is defined for a function on a $[-1,1]^D$ domain. For an arbitrary hyperrectangular $\mathcal{X}=[\underline{x}_{1},\overline{x}_1]\times\ldots \times[\underline{x}_D,\overline{x}_D]$, we introduce a linear transformation $\tau_{\mathcal{X}}:[-1,1]^D\rightarrow\mathcal{X}$ componentwise defined by
\begin{align}
\tau_{\mathcal{X}}(z_i)=\overline{x}_i+0.5(\underline{x}_i-\overline{x}_i)(1-z_i).\label{Transformation}
\end{align}
Let $\overline{N}:=(N_1,\ldots,N_D)$ with $N_i \in\nn_0$ for $i=1,\ldots,D$. We define the index set
\begin{align*}
\mathcal{J}:=\{j\in\mathbb{N}^D:1\le j_i\le N_D\ \text{for}\ i=1,\ldots,d\}.
\end{align*}
The Chebyshev polynomials are defined for $z\in[-1,1]^D$ and $j\in\mathcal{J}$ by
\begin{align*}
T_{j}(z) = \prod_{i=1}^D T_{j_i}(z_i),\quad T_{j_i}(z_i)=\cos(j_i\cdot\text{acos}(z_i)),
\end{align*}
and the $j$-th Chebyshev polynomial on $\mathcal{X}$ as $p_j(x)=T_j(\tau^{-1}_{\mathcal{X}}(x))1_{\mathcal{X}}(x)$. The Chebyshev points are given by 
\begin{align*}
z^k = (z_{k_1},\dots,z_{k_D}), \ z_{k_i}=\cos\left(\pi\frac{k_i}{N_i}\right)\text{ for }k_i=0,\ldots,N_i\text{ and }i=1,\ldots,D.
\end{align*}
and the transformed Chebyshev points by $x^k=\tau_{\mathcal{X}}(z^k).$ The Chebyshev interpolation of a function $f:\mathcal{X}\rightarrow\mathbb{R}$ with $\prod_{i=1}^D (N_{i}+1)$ summands can be written as a sum of Chebyshev polynomials
\begin{align}\label{Cheby_Interpolation}
I_{\overline{N}}(f)(x)=\sum_{j\in \mathcal{J}} c_{j} T_{j}(\tau_{\mathcal{X}}^{-1}(x))=\sum_{j\in \mathcal{J}} c_{j}p_{j}(x) \qquad \text{for}\quad x\in\mathcal{X}
\end{align}
with coefficients $c_j$ for $j\in \mathcal{J}$ 
\begin{align}\label{def:Chebycj}
c_j&=\Big( \prod_{i=1}^D \frac{2^{\1_{\{0<j_i<N_i\}}}}{N_i}\Big)\sum_{k\in\mathcal{J}}{}^{''}f(x^k)T_{j}(z^k)
\end{align}
where $\sum{}^{''}$ indicates the summand is multiplied by $1/2$ if $k_i=0$ or $k_i=N_i$.

\subsection{The Dynamic Chebyshev method}
In this section, we present the new approach to solve a dynamic programming problem via backward induction using Chebyshev polynomial interpolation.
\begin{defin}\label{defin_DPP}
We consider a Dynamic Programming Problem (DPP) with value function
\begin{align}
V_{T}(x)&=g(T,x)\label{DPP_1}\\
V_{t_u}(x)&=f\left(g(t_u,x), \EE[V_{t_{u+1}}(X_{t_{u+1}})\vert X_{t_u}=x]\right)\label{DPP_2},
\end{align}
where $t=t_0<\ldots<t_{n}=T$ and $f:\mathbb{R}\times\mathbb{R}\rightarrow\mathbb{R}$ is Lipschitz continuous with constant $L_f$. 
\end{defin}

At the initial time $T=t_{n}$, we apply Chebyshev interpolation to the function $g(T,x)$, i.e. for $x\in\mathcal{X}$,
\begin{align*}
V_{T}(x)=g(T,x)\approx \sum_{j\in \mathcal{J}} c_j(T)p_j(x)=:\widehat{V}_{T}(x)
\end{align*}
At the first time step $t_{n-1}$, the derivation of $\EE[g(t_n,X_{t_{n}})\vert X_{t_{n-1}}=x]$ is replaced by $\EE[\sum_{j} c_j(t_n)p_j(X_{t_{n}})\vert X_{t_{n-1}}=x]$ yielding
\begin{align*}
V_{t_{n-1}}(x)&=f\left(g(t_{n-1},x),\EE[V_{t_n}(X_{t_{n}})\vert X_{t_{n-1}}=x]\right)\\
&\approx f\Big(g(t_{n-1},x),\EE\Big[ \sum_{j\in \mathcal{J}} c_j(t_n)p_j(X_{t_{n}})\Big\vert X_{t_{n-1}}=x\Big]\Big)\\
&= f\Big(g(t_{n-1},x), \sum_{j\in \mathcal{J}} c_j(t_n)\EE\Big[p_j(X_{t_{n}})\Big\vert X_{t_{n-1}}=x\Big]\Big).
\end{align*}
At time step $t_{n-1}$ the value function $V_{t_{n-1}}$ needs only to be evaluated at the specific Chebyshev nodes. Hence, denoting with $x^k=(x_{k_1},\ldots,x_{k_D})$ the Chebyshev nodes, it suffices to evaluate
\begin{align}
V_{t_{n-1}}(x^k)\approx f\Big(g(t_{n-1},x^k), \sum_{j\in \mathcal{J}} c_j(t_n)\EE\Big[p_j(X_{t_{n}})\Big\vert X_{t_{n-1}}=x^k\Big]\Big)=:\widehat{V}_{t_{n-1}}(x^k).
\end{align}
A linear transformation of $(\widehat{V}_{t_{n-1}}(x^k))_{k\in\mathcal{J}}$ yields the Chebyshev coefficients according to \eqref{def:Chebycj} which determines the Chebyshev interpolation $\widehat{V}_{t_{n-1}}=\sum_j c_j(t_{n-1})p_j$. We apply this procedure iteratively as described in detail in Algorithm \ref{Algorithm_General}.

The stochastic part is gathered in the expectations of the Chebyshev polynomials conditioned on the Chebyshev nodes, i.e. $\Gamma_{j,k}(t_u)=\EE[p_j(X_{t_{u+1}})\vert X_{t_u}=x^k]$. Moreover, if an equidistant time stepping is applied the computation can be further simplified. If for the underlying stochastic process
\begin{align}\label{eq:generalized_moments}
\Gamma_{j,k}(t_u)=\EE[p_j(X_{t_{u+1}})\vert X_{t_u}=x^k]=\EE[p_j(X_{t_{1}})\vert X_{t_0}=x^k]=:\Gamma_{j,k} 
\end{align}
for $u=0,\ldots,n-1$, then the conditional expectations need to be computed only for one time step, see Algorithm \ref{Algorithm_Stationary}. One can pre-compute these conditional expectations and thus, the method allows for an offline/online decomposition.



\begin{algorithm}[H]
\caption{Dynamic Chebyshev algorithm}\label{Algorithm_General}
\begin{algorithmic}[1]
\Require $\ol{N}\in\mathbb{N}^{D}$, $\mathcal{X}=[\underline{x}_1,\overline{x}_1]\times\ldots\times[\underline{x}_D,\overline{x}_D]$, $0=t_0,\ldots,t_{n}=T$
\State Determine index set $\mathcal{J}$ and nodal points $x^k=(x_{k_1},\ldots,x_{k_D})$
\vspace{0.1cm}
\State \textbf{Pre-computation step:}
\State \quad\ For all $j,k\in \mathcal{J}$ and all $t_u,\ u=0,\ldots,n-1$
\State \quad\ Compute $\Gamma_{j,k}(t_u)=\EE[p_j(X_{t_{u+1}})\vert X_{t_{u}}=x^k]$
\vspace{0.1cm}
\State \textbf{Time $T$}
\State \quad\ $\widehat{V}_T(x^k)=g(T,x^k),\ k\in\mathcal{J}$, derive
\State \quad\ $c_j(T) = D_{\ol{N}}(j)\sum_{k\in \mathcal{J}}\!{}^{''}\widehat{V}_{T}(x^{k})T_{j}(z^{k})$
\State \quad\ Obtain Chebyshev interpolation $\widehat{V}_{T}(x)=\sum_{j\in \mathcal{J}} c_j(T)p_j(x)$ of $V_{T}(x)$
\vspace{0.1cm}
\State \textbf{Iterative time stepping} from $t_{u+1}\rightarrow t_{u},\ u=n-1,\ldots,1$
\State \quad\ Given Chebyshev interpolation of $\widehat{V}_{t_{u+1}}(x)=\sum_{j\in \mathcal{J}} c_j(t_{u+1})p_j(x)$
\State \quad\ Derivation of $\widehat{V}_{t_{u}}(x^k),\ k\in\mathcal{J}$ at the nodal points
\State \quad\ $\widehat{V}_{t_{u}}(x^k)=f(g(t_{u},x^k),\sum_{j\in \mathcal{J}} c_j(t_{u+1})\Gamma_{j,k}(t_u) )$
\State \quad\ Derive $c_j(t_{u}) = D_{\ol{N}}(j)\sum_{k\in \mathcal{J}}\!{}^{''}\widehat{V}_{t_{u}}(x^{k})T_{j}(z^{k})$
\State \quad\ Obtain Chebyshev interpolation $\widehat{V}_{t_{u}}(x)=\sum_{j\in \mathcal{J}} c_j(t_{u})p_j(x)$ of $V_{t_{t}}(x)$
\vspace{0.1cm}
\State \textbf{Deriving the solution at $t=0$}
\State \quad\ $\widehat{V}_{0}(x)=\sum_{j\in \mathcal{J}} c_j(0)p_j(x)$
\end{algorithmic}
\end{algorithm}
\vspace{-0.2cm}
\begin{algorithm}
\caption{Simplified Dynamic Chebyshev algrithm}\label{Algorithm_Stationary}
\begin{algorithmic}[1]
\Require Time steps $0=t_1,\ldots,t_{n}=T$ with $\Delta t:=t_{u}-t_{u-1}$
\State Replace in Algorithm \ref{Algorithm_General} Lines 2-4 with:
\vspace{0.1cm}
\State \textbf{Pre-computation step:}
\State \quad\ Compute $\Gamma_{j,k}=\EE[p_j(X_{\Delta  t})\vert X_{0}=x^k]$ for all $j, k\in\mathcal{J}$
\end{algorithmic}
\end{algorithm}
\FloatBarrier
\vspace{-0.5cm}
\section{Error Analysis}
In this section we analyse the error of Algorithm \ref{Algorithm_General}, i.e.
\begin{align}\label{eq:error_deterministic}
\varepsilon_{t_{u}}:=\max_{x\in\mathcal{X}}\vert V_{t_{u}}(x)-\widehat{V}_{t_{u}}(x)\vert.
\end{align}
Two different error sources occur at $t_u$, the classical interpolation error of the Chebyshev interpolation and a distortion error at the nodal points. The latter comes from the fact that the values $\widehat{V}_{t_u}(x^k)$ are approximations of $V_{t_u}(x^k)$. The behaviour of the interpolation error depends on the regularity of the value function. Here, we assume analyticity of the value function. The concept can be extended to further cases such as assuming differentiability or piecewise analyticity. The latter is discussed in preliminary form in \cite[Section 5.3]{Mahlstedt2017} and is further investigated in a follow-up paper. Hence, we need a convergence result for the Chebyshev interpolation which incorporates a distortion error at the nodal points.\\

First, we introduce the required notation. A \textit{Bernstein ellipse} $\mathcal{B}([-1,1],\varrho)$ with $\varrho>1$ is defined as the open region in the complex plane bounded by an ellipse with foci $\pm 1$ and semiminor and semimajor axis lengths summing to $\varrho$. We define a \textit{generalized Bernstein ellipse} $\mathcal{B}(\mathcal{X},\varrho)$ around the hyperrectangle $\mathcal{X}$ with parameter vector $\varrho\in(1,\infty)^D$ as
\begin{align*}
\mathcal{B}(\mathcal{X},\varrho):=\mathcal{B}([\ul{x}_1,\ol{x}_1],\varrho_1)\times\ldots
\times \mathcal{B}([\ul{x}_D,\ol{x}_D],\varrho_D )
\end{align*}
with $\mathcal{B}([\ul{x},\ol{x}],\varrho):=\tau_{[\ul{x},\ol{x}]}\circ \mathcal{B}([-1,1],\varrho)$,
where for $x\in\cc$ we have the transform $\tau_{[\ul{x},\ol{x}]}\big(\Re(x)\big)
:=\ol{x} + \frac{\ul{x}-\ol{x}}{2}\big(1-\Re(x)\big)$ and
$\tau_{[\ul{x},\ol{x}]}\big(\Im(x)\big):= \frac{\ol{x}-\ul{x}}{2}\Im(x)$ where the sets
$\mathcal{B}([-1,1],\varrho_i)$ are Bernstein ellipses for $i=1,\ldots,D$.

\begin{proposition}\label{Cheby_Err_Distortion_multi}
Let $\mathcal{X}\ni x\mapsto f(x)$ be a real-valued function with an analytic extension to some generalized Bernstein ellipse $\mathcal{B}(\mathcal{X},\varrho)$ for $\varrho\in(1,\infty)^{D}$ with $\sup_{x\in \mathcal{B}(\mathcal{X},\varrho)}|f(x)|\leq b$. Assume distorted values $f^{\varepsilon}(x^{k})=f(x^{k})+\varepsilon(x^{k})$ with $\vert\varepsilon(x^{k})\vert\le \ol{\varepsilon}$ at all nodes $x^{k}$. Then
\begin{align*}
\max_{x\in\mathcal{X}}\big|f(x) - I_{\ol{N}}(f^{\varepsilon})(x)\big|\leq \varepsilon_{int}(\varrho,N,D,B) + \ol{\varepsilon}\Lambda_{\ol{N}}.
\end{align*}
with
\begin{align}\label{eq:Cheby_error_bound_alpha}
\varepsilon_{int}(\varrho,N,D,B):= 2^{\frac{D}{2}+1}\cdot B \cdot\bigg(\sum_{i=1}^D\varrho_i^{-2N_i}\prod_{j=1}^D\frac{1}{1-\varrho_j^{-2}}\bigg)^{\frac{1}{2}}
\end{align}
and Lebesgue constant $\Lambda_{\ol{N}}\leq \prod_{i=1}^{D}\big(\frac{2}{\pi}\log(N_i+1)+1\big)$.
\end{proposition}

\begin{proof}
Using the linearity of the interpolation operator we obtain for the Chebyshev interpolation of $f^{\varepsilon}$ with $f^{\varepsilon}(x^{k})=f(x^{k})+\varepsilon(x^{k})$ that
\begin{align*}
I_{\ol{N}}(f^{\varepsilon})(x)=I_{\ol{N}}(f)(x)+I_{\ol{N}}(\varepsilon)(x).
\end{align*}
The tensor-based multivariate Chebyshev interpolation $I_{\ol{N}}(\varepsilon)$ can be written in Lagrange form
\begin{align*}
I_{\ol{N}}(\varepsilon)(x)=\sum_{j\in\mathcal{J}}\varepsilon(x^j)\lambda^{j}(x) \qquad \text{with} \quad \lambda^{j}(x)=\prod_{i=1}^{D}\ell_{j_i}(\tau^{-1}_{[\ul{x}_i,\ol{x}_i]}(x_i))
\end{align*}
where $\ell_{j_i}(z)=\prod_{k\neq j_i}\frac{z-z_{k}}{z_{j_i}-z_{k}}$ is the $j_{i}-$th Lagrange polynomial. This yields 
\begin{align*}
\max_{x\in\mathcal{X}}\left\vert I_{N}(\varepsilon)(x) \right\vert
= \max_{x\in\mathcal{X}}\Big\vert \sum_{j\in\mathcal{J}}\varepsilon(x^j)\lambda^{j}(x) \Big\vert
\leq \ol{\varepsilon}\max_{x\in\mathcal{X}}\sum_{j\in\mathcal{J}}|\lambda^{j}(x)|
=:\ol{\varepsilon} \Lambda_{\ol{N}}.
\end{align*} 
The term $\Lambda_{\ol{N}}$ is the Lebesgue constant of the (multivariate) Chebyshev nodes which is given by
\begin{align*}
\Lambda_{\ol{N}}=\max_{x\in\mathcal{X}}\sum_{j\in\mathcal{J}}\big|\lambda^{j}(x)\big|
=\max_{x\in\mathcal{X}}\sum_{j\in\mathcal{J}}\prod_{i=1}^D \big|\ell_{j_i}(x_i)\big|
=\prod_{i=1}^D \max_{x_i\in[\ul{x}_i,\ol{x}_i]} \sum_{j_i=0}^{N_i}
\Big|\ell_{j_i}\big(\tau^{-1}_{[\ul{x}_i,\ol{x}_i]}(x_i)\big)\Big|.
\end{align*}
Since $\max_{x_i\in[\ul{x}_i,\ol{x}_i]}\sum_{j_i=0}^{N_i} |\ell_{j_i}(\tau^{-1}_{[\ul{x}_i,\ol{x}_i]}
(x_i))|= \max_{z\in[-1,1]}\sum_{j_i=0}^{N_i} |\ell_{j_i}(z)| =\Lambda_{N_i}$, which is the Lebesgue constant of the univariate Chebyshev interpolation, we have $\Lambda_{\ol{N}}=\prod_{i=1}^{D}\Lambda_{N_i}$. From \cite[Theorem 15.2]{Trefethen2013} we obtain for the univariate Chebyshev interpolation $\Lambda_{N}\leq \frac{2}{\pi}\log(N+1)+1$ and hence
\begin{align}\label{eq:multi_Lesbegue_constant}
\Lambda_{\ol{N}}\leq \prod_{i=1}^{D}\Big(\frac{2}{\pi}\log(N_i+1)+1\Big).
\end{align}

For the distorted Chebyshev interpolation holds
\begin{align*}
\big|f(x) - I_{\ol{N}}(f^{\varepsilon})(x)\big|\leq |f(x) - I_{\ol{N}}(f)(x)\big|+|I_{\ol{N}}(\varepsilon)(x)\big|.
\end{align*}
Therefore, the proposition follows directly from \eqref{eq:multi_Lesbegue_constant} and \cite{SauterSchwab2010}.
\end{proof}

We use this result to investigate the error of the Dynamic Chebyshev method. First, we introduce the following assumption.
\begin{assumptions}\label{assumption_analytic_value}
We assume $\mathcal{X}\ni x\mapsto V_{t_u}(x)$ is a real valued function that has an analytic extension to a generalized Bernstein ellipse $\mathcal{B}(\mathcal{X},\varrho_{t_u})$ with $\varrho_{t_u}\in (1,\infty)^D$ and $\sup_{x\in \mathcal{B}(\mathcal{X},\varrho_{t_u})}|V_{t_u}(x)|\le B_{t_u}$ for $u=1,\ldots,n$.
\end{assumptions}
Proposition \ref{Cond_analy_DPP_nosplit} provides conditions on the process $X$ and the functions $f$ and $g$ that guaranty Assumptions \ref{assumption_analytic_value}. Under this assumptions, we can apply Proposition \ref{Cheby_Err_Distortion_multi} to obtain an error bound for the Dynamic Chebyshev method at each time step. This error bound has a recursive structure, since the values of $V_{t_{u}}$ depend on the conditional expectation of $V_{t_{u+1}}$. The interpolation error of the final time step is of form \eqref{eq:Cheby_error_bound_alpha}. At any other time step $t_{u}$ an additional distortion error by approximating the function values at the nodal points by
\begin{align*}
V_{t_{u}}(x^k)\approx f\bigg(g(t_{u},x^k), \sum_{j\in \mathcal{J}} c_j(t_{u+1})\EE[p_j(X_{t_{u+1}})\vert X_{t_{u}}=x^k]\bigg)=\widehat{V}_{t_{u}}(x^k)
\end{align*}
comes into play. Proposition \ref{Cheby_Err_Distortion_multi} yields
\begin{align*}
\varepsilon_{t_{u}}:=\max_{x\in\mathcal{X}}\vert V_{t_{u}}(x)-\widehat{V}_{t_{u}}(x)\vert \le \varepsilon_{int}(\varrho_{t_{u}},N,D,B_{t_{u}}) +\Lambda_{\ol{N}}F_{t_{u}},
\end{align*}
where $F_{t_{u}}:=\max_{j\in \mathcal{J}}\vert V_{t_{u}}(x_{j})- \widehat{V}_{t_{u}}(x_{j})\vert$. The term $F_{t_{u}}$ depends on the function $f$ and the interpolation error at the previous time step $t_{u+1}$.

Moreover, two additional error sources can influence the error bound. If there is no closed-form solution for the generalized moments $\EE[p_j(X_{t_{u+1}})\vert X_{t_{u}}=x^k]$ a numerical technique, e.g. numerical quadrature or Monte Carlo methods, introduces an additional error. The former is typically deterministic and bounded whereas the latter is stochastic. In order to incorporate this error in the following error analysis we introduce some additional notation. The conditional expectation can be seen as a linear operator which operates on the vector space of all continuous functions $\mathcal{C}(\mathbb{R}^{D})$ with finite $L^{\infty}$-norm
\begin{align*}
\Gamma_{t_{u}}^{k}:\mathcal{C}(\mathbb{R}^{D})\rightarrow\mathbb{R} \quad \text{with} \quad \Gamma_{t_{u}}^{k}(f):=\EE[f(X_{t_{u+1}})\vert X_{t_{u}}=x^k].
\end{align*}
Define the subspace of all $D$ variate polynomials $\mathcal{P}_{\ol{N}}(\mathcal{X}):=\text{span}\{p_{j},\ j\in\mathcal{J}\}$ equipped with the $L^{\infty}$-norm. 
We assume the operator $\Gamma_{t_{u}}^{k}$ is approximated by a linear operator $\widehat{\Gamma}_{t_{u}}^{k}:\mathcal{P}_{\ol{N}}(\mathcal{X})\rightarrow\mathbb{R}$ on $\mathcal{P}_{\ol{N}}(\mathcal{X})$ which fullfills one of the two following conditions. For all $u=0,\ldots,n$ the approximation is either deterministic and the error is bounded by a constant $\ol{\delta}$,
\begin{align}\label{cond_exp_delta_bound}
\vert\vert\Gamma_{t_{u}}^{k}-\widehat{\Gamma}_{t_{u}}^{k}\vert\vert_{op}:=\sup_{\substack{p\in\mathcal{P}_{\ol{N}}\\ \vert\vert p\vert\vert=1}}\left\vert \Gamma_{t_{u}}^{k}(p)-\widehat{\Gamma}_{t_{u}}^{k}(p)\right\vert\leq \ol{\delta} \quad \forall k\in\mathcal{J}
\tag{GM}
\end{align}
or the approximation is stochastic and uses $M$ samples of the underlying process and the polynomials $p$ may have stochastic coefficients. In this case we assume the error bound
\begin{align}\label{cond_exp_MC_bound}
\vert\vert\Gamma_{t_{u}}^{k}-\widehat{\Gamma}_{t_{u}}^{k}\vert\vert_{op}:=\sup_{\substack{p\in\mathcal{P}_{\ol{N}}\\ \vert\vert p\vert\vert_{\infty}^{\star}=1}}\EE\left[\left\vert \Gamma_{t_{u}}^{k}(p)-\widehat{\Gamma}_{t_{u}}^{k}(p)\right\vert\right]\leq \delta^{\star}(M) \quad \forall k\in\mathcal{J}
\tag{$\text{GM}^{*}$}
\end{align}
with norm $\vert\vert p\vert\vert_{\infty}^{\star}=\max_{x\in\mathcal{X}}\EE[\vert p(x)\vert]$. In order to incorporate stochasticity of $\widehat{V}_{t_u}(x)$, we replace \eqref{eq:error_deterministic} by
\begin{align}\label{eq:error_stochastic}
\varepsilon_{t_{u+1}}:=\max_{x\in\mathcal{X}}\EE\left[\left\vert V_{t_{u}}(x)-\widehat{V}_{t_{u}}(x)\right\vert\right].
\end{align}
Note that in the deterministic case \eqref{eq:error_deterministic} and \eqref{eq:error_stochastic} coincide. Additionally, a truncation error is introduced by restricting to a compact interpolation domain $\mathcal{X}$. We assume that the conditional expectation of the value function outside this set is bounded by a constant
\begin{align}\label{trunc_err_bound}
\EE[V_{t_{u+1}}(X_{t_{u+1}})\1_{\mathbb{R}^{D}\setminus\mathcal{X}}\vert X_{t_{u}}=x^{k}]\leq \varepsilon_{tr}.
\tag{TR}
\end{align}

The following theorem provides an error bound for the Dynamic Chebyshev method. 

\begin{theorem}\label{Error_DPP_Algo1_Lip}
Let the DPP be given as in Definition \ref{defin_DPP}. Assume the regularity Assumptions \ref{assumption_analytic_value} hold and the boundedness of the truncation error \eqref{trunc_err_bound}. Then we have 
\begin{align}\label{Error_t_k_Lip_delta}
\varepsilon_{t_{u}}\leq \sum_{j=u}^{n}C^{j-u}\varepsilon_{int}^{j}\ +\ \Lambda_{\ol{N}}L_{f}\sum_{j=u+1}^{n}C^{j-(u+1)}(\varepsilon_{tr}+\varepsilon_{gm}\ol{V}_{j})
\end{align}
with with $\varepsilon_{gm}=\ol{\delta}$ if assumption \eqref{cond_exp_delta_bound} holds and $\varepsilon_{gm}=\delta^{\star}(M)$ if assumption \eqref{cond_exp_MC_bound} holds 
and $C=\Lambda_{\ol{N}}L_{f}(1+\varepsilon_{gm})$, $\ol{V}_{j}=\max_{x\in\mathcal{X}}|V_{t_{j}}(x)|$ and $\varepsilon_{int}^{j}=\varepsilon_{int}
(\varrho_{t_{j}},N,D,B_{t_{j}})$.
\end{theorem}
\begin{proof}
The proof of the theorem can be found in the appendix.
\end{proof}

The following corollary provides a simplified version of the error bound \eqref{Error_t_k_Lip_delta} presenting its decomposition into three different error sources (interpolation error $\varepsilon_{int}$, truncation error $\varepsilon_{tr}$ and the error from the numerical computation of the generalized moments $\varepsilon_{gm}$).
\begin{corollary}\label{Error_DPP_Algo1_Lip_no}
Let the setting be as in Theorem \ref{Error_DPP_Algo1_Lip}. Then the error is bounded by
\begin{align}\label{eq:error_bound_simplified}
\varepsilon_{t_{u}}&\leq \left(\varepsilon_{int}(\ul{\varrho},N,D,\ol{B})+\varepsilon_{tr}+\varepsilon_{gm}\ol{V}\right)\tilde{C}^{n+1-u}
\end{align}
with $\tilde{C}=\max\{2,C\}$, $\ul{\varrho}=\min_{1\leq u\leq n}\varrho_{t_u}$, $\ol{B}=\max_{1\leq u\leq n}B_{t_{u}}$, $\ol{V}=\max_{u\leq j\leq n}\ol{V}_j$.

Moreover, if $\varepsilon_{tr}=0$, $L_{f}=1$ and $N=N_{i}$, $i=1,\ldots,D$ the error bound can be simplified further. Under \eqref{cond_exp_MC_bound} $\delta^{\star}(M)\leq c/\sqrt{M}$, $c>0$ yields
\begin{align*}
\varepsilon_{t_{u}}\leq \tilde{c}_{1}\varrho^{-N}\log(N)^{D\,n}\ +\ \tilde{c}_{2}\log(N)^{D\,n}M^{-0.5}.
\end{align*}
for some constants $\tilde{c}_{1},\tilde{c}_{2}>0$. Under \eqref{cond_exp_delta_bound} the term $M^{-0.5}$ is replaced by $\ol{\delta}$.
\end{corollary}

\begin{proof}
Assuming $C>2$ and using the geometric series, the first term in the error bound \eqref{Error_t_k_Lip_delta} can be rewritten as
\begin{align*}
\sum_{j=u}^{n}C^{j-u}\varepsilon_{int}^j\leq \ol{\varepsilon}_{int}\sum_{j=u}^{n}C^{j-u} =\ol{\varepsilon}_{int}\sum_{k=0}^{n-u}C^{k}=\ol{\varepsilon}_{int}\left(\frac{1-C^{n+1-u}}{1-C}\right)\leq\ol{\varepsilon}_{int}\,C^{n+1-u},
\end{align*}
where $\ol{\varepsilon}_{int}=\max_j \varepsilon_{int}^j=\max_j \varepsilon_{int}(\varrho_{t_{j}},N,D,B_{t_{j}})\leq\varepsilon_{int}(\ul{\varrho},N,D,\ol{B})$ for $\ul{\varrho}=\min_{1\leq u\leq n}\varrho_{u}$ and $\ol{B}=\max_{1\leq u\leq n}B_{t_{u}}$. For $C\leq 2$ the sum is bounded by $\ol{\varepsilon}_{int}\,2^{n+1-u}$. Similar, we obtain for the second term in the error bound \eqref{Error_t_k_Lip_delta} with $\beta=(\varepsilon_{tr}+\varepsilon_{gm}\ol{V}_{j})$
\begin{align*}
\Lambda_{\ol{N}}L_{f}\sum_{j=u+1}^{n}C^{j-(u+1)}\beta_{j}\leq \Lambda_{\ol{N}}L_{f}\,\ol{\beta}\sum_{k=0}^{n-(u+1)}C^{k}
\leq \Lambda_{\ol{N}}L_{f}\,\ol{\beta} C^{n-u}\leq \ol{\beta} C^{n+1-u}
\end{align*}
where $\ol{\beta}=\max_j\beta_j$. Moreover, we used $\Lambda_{\ol{N}}L_{f}\leq \Lambda_{\ol{N}}L_{f}(1+\varepsilon_{gm})=C$ in the last step. Thus, we obtain the following error bound \eqref{Error_t_k_Lip_delta}
\begin{align*}
\varepsilon_{t_{u}}&\leq (\ol{\varepsilon}_{int}\ +\ \ol{\beta})\,\tilde{C}^{n+1-u}= \left(\varepsilon_{int}(\ul{\varrho},N,D,\ol{B})+\varepsilon_{tr}+\varepsilon_{gm}\ol{V}\right)\tilde{C}^{n+1-u},
\end{align*}
where $\tilde{C}=\max\{2,C\}$ and $\ol{V}=\max_j\ol{V}_j$, which shows \eqref{eq:error_bound_simplified}.\\

Furthermore, using the definition of the error bound \eqref{eq:Cheby_error_bound_alpha} and $N=N_{i}$, $i=1,\ldots,D$ we conclude that $\varepsilon_{int}(\ul{\varrho},N,D,\ol{B})\leq c_{1}\varrho^{-N}$ for a constant $c_{1}>0$. For the Lebesgue constant of the Chebyshev interpolation exists a constant $c_{2}>0$ such that
\begin{align*}
\Lambda_{\ol{N}}\leq \prod_{i=1}^{D}\big(\frac{2}{\pi}\log(N+1)+1\big)
\leq \prod_{i=1}^{D}\big(\frac{4}{\pi}+1\big)\log(N)
\leq c_{2}\log(N)^{D}.
\end{align*}
Under \eqref{cond_exp_MC_bound}, $\delta^{\star}(M)\leq c/\sqrt{M}$, $c>0$ yields with $\varepsilon_{tr}=0$, $L_{f}=1$
\begin{align*}
\varepsilon_{t_{u}}&\leq \left(\varepsilon_{int}(\ul{\varrho},N,D,\ol{B})+\varepsilon_{tr}+\varepsilon_{gm}\ol{V}\right)\left(\Lambda_{\ol{N}}L_{f}(1+\varepsilon_{gm})\right)^{n+1-u}\\
&\leq \left(c_{1}\varrho^{-N}+c\ol{V}M^{-0.5}\right)\left(c_{2}\log(N)^{D}(1+cM^{-0.5})\right)^{n}\\
&\leq \tilde{c}_{1}\varrho^{-N}\log(N)^{D\,n}\ +\ \tilde{c}_{2}\log(N)^{D\,n}M^{-0.5}
\end{align*}
and this converges towards zero for $N\rightarrow\infty$ if $\sqrt{M}>\log(N)^{D\,n}$. If \eqref{cond_exp_delta_bound} holds we have $\varepsilon_{gm}=\ol{\delta}$ and the term $M^{-0.5}$ is replaced by $\ol{\delta}$.
\end{proof}

The following proposition provides conditions under which the value function has an analytic extension to some generalized Bernstein ellipse and Assumptions \ref{assumption_analytic_value} hold.
\begin{proposition}\label{Cond_analy_DPP_nosplit}
Consider a DPP as defined in \eqref{DPP_1} and \eqref{DPP_2} with equidistant time-stepping and $g(t,x)=g(x)$. Let $X=(X_{t})_{0\leq t\leq T}$ be a Markov process with stationary increments. Assume $e^{\langle \eta,\cdot\rangle}g(\cdot)\in L^{1}(\mathbb{R}^{D})$ for some $\eta\in\mathbb{R}^{D}
$ and $g$ has an analytic extension to the generalized Bernstein ellipse $\mathcal{B}(\mathcal{X},\varrho_{g})$. Furthermore, assume $f:\mathbb{R}\times\mathbb{R}\rightarrow\mathbb{R}$ has an analytic extension to $\mathbb{C}^{2}$. If
\begin{itemize}
\item[(i)] the characteristic function $\varphi^{x}$ of $X_{\Delta t}$ with $X_{0}=x$ is in $L^{1}(\mathbb{R}^{D})$ for every $x\in\mathcal{X}$,
\item[(ii)] for every $z\in\mathbb{R}^{D}$ the mapping $x\mapsto\varphi^{x}(z-i\eta)$ has an analytic extension to $B(\mathcal{X},\varrho_{\varphi})$ and there are constants $\alpha\in(1,2]$ and $c_{1},c_{2}>0$ such that $\sup_{x\in B(\mathcal{X},\varrho_{\varphi})}|\varphi^{x}(z)|\leq c_{1}e^{-c_{2}|z|^{\alpha}}$ for all $z\in\mathbb{R}^{D}$,
\end{itemize}
then the value function $x\mapsto V_{t_{u}}(x)$ of the DPP has an analytic extension to $B(\mathcal{X},\varrho)$ with $\varrho=\varrho_{g}$.
\end{proposition}
\begin{proof}
At $T$ the value function $x\mapsto V_{T}(x)$ is analytic since $V_{T}(x)=g(x)$ and $g$ has an analytic extension by assumption. Moreover, $e^{\langle \eta,\cdot\rangle}g(\cdot)\in L^{1}(\mathbb{R}^{D})$ for some $\eta\in\mathbb{R}^{D}$. We assume $e^{\langle \eta,\cdot\rangle}V_{t_{u+1}}(\cdot)\in L^{1}(\mathbb{R}^{D})$ and $V_{t_{u+1}}$ has an analytic extension to $B(\mathcal{X},\varrho)$. Then the function
\begin{align*}
x\mapsto V_{t_u}(x)=f\left(g(x), \EE[V_{t_{u+1}}(X_{t_{u+1}})\vert X_{t_u}=x]\right)
\end{align*}
is analytic if $x\mapsto \EE[V_{t_{u+1}}(X_{t_{u+1}})\vert X_{t_u}=x]=\EE[V_{t_{u+1}}(X_{\Delta t}^{x})]$ has an analytic extension. From \cite[Conditions 3.1]{GassGlauMahlstedtMair2018} we obtain conditions (A1)-(A4) under which a function of the form $(p^{1},p^{2})\mapsto\EE[f^{p^{1}}(X^{p^{2}})]$ is analytic. In our case we only have the parameter $p^{2}=x$ and so $X^{p^{2}}=X_{\Delta t}^{x}$. 
Condition (A1) is satisfied since $e^{\langle \eta,\cdot\rangle}V_{t_{u+1}}(\cdot)\in L^{1}(\mathbb{R}^{D})$ and for (A2) we have to verify that $|\widehat{V}_{t_{u+1}}(-z-i\eta)|\leq c_{1}e^{c_{2}|z|}$ for constants $c_{1},c_{2}>0$. 
\begin{align*}
|\widehat{V}_{t_{u+1}}(-z-i\eta)|&=\Big|\int\limits_{\mathbb{R}^{D}}e^{i\langle y,-z-i\eta\rangle}V_{t_{u+1}}\text{d}y\Big|\\
&\leq\int\limits_{\mathbb{R}^{D}}\vert e^{-i\langle y,z\rangle}\vert \left|e^{\langle y,\eta\rangle}V_{t_{u+1}}(y)\right|\text{d}y\\
&\leq \vert\vert e^{\langle \eta,\cdot\rangle}V_{t_{u+1}}(\cdot)\vert\vert_{L^{1}}
\end{align*}
and thus (A2) holds. The remaining conditions (A3)-(A4)  are equivalent to our conditions (i)-(ii) and \cite[Theorem 3.2]{GassGlauMahlstedtMair2018} yields the analyticity of $x\mapsto \EE[V_{t_{u+1}}(X_{\Delta t}^{x})]$ on the Bernstein ellipse $\mathcal{B}(\mathcal{X},\varrho_{\varphi})$. Hence, $x\mapsto V_{t_{u}}(x)$ is a composition of analytic functions and therefore analytic on the intersection of the domains of analyticity $\mathcal{B}(\mathcal{X},\varrho_{\varphi})\cap \mathcal{B}(\mathcal{X},\varrho_{g})=\mathcal{B}(\mathcal{X},\varrho)$ with $\varrho=\min\{\varrho_{g},\varrho_{\varphi}\}$.\\

It remains to prove that $e^{\langle \eta,\cdot\rangle}V_{t_{u}}(\cdot)\in L^{1}(\mathbb{R}^{D})$. Here the Lipschitz continuity of $f$ yields
\begin{align*}
\vert\vert e^{\langle \eta,\cdot\rangle}V_{t_{u}}(\cdot)\vert\vert_{L^{1}}\leq L_{f}\left(\vert\vert e^{\langle \eta,\cdot\rangle}g(\cdot)\vert\vert_{L^{1}} + \vert\vert e^{\langle \eta,\cdot\rangle}V_{t_{u+1}}(\cdot)\vert\vert_{L^{1}}\right)<\infty.
\end{align*}
\end{proof}

Often, the discrete time problem \eqref{DPP_1} and \eqref{DPP_2} is an approximation of a continuous time problem and thus, we are interested in the error behaviour for $n\rightarrow\infty$. 
\begin{remark}
Assume the setup of Corollary \ref{Error_DPP_Algo1_Lip_no}. Moreover, assume that $\varepsilon_{tr}=\varepsilon_{gm}=0$. If we let $N$ and $n$ go to infinity, we have to ensure that the error bound tends to zero. We use that $\varepsilon_{int}(\ul{\varrho},N,D,\ol{B})\leq C_{1}\varrho^{-\ul{N}}$ for a constant $C_1>0$ and $\ul{N}=\min_i N_i$. The following condition on the relation between $n$ and $N$ ensures convergence
\begin{align*}
n < \frac{\log(\varrho)}{C_{1}D}\cdot\frac{\ul{N}}{\log(\Lambda_{\ol{N}})+\log(L_f)}\,+\,1.
\end{align*}
\end{remark}

\section{Implementational aspects of the Dynamic Chebyshev method}
In this section we discuss several approaches to compute the generalized moments \eqref{eq:generalized_moments} which contain the model dependent part. Moreover, preparing the numerical experiments we tailor the Dynamic Chebyshev method to the pricing of American put options.

\subsection{Derivation of generalized moments}
Naturally, the question arises how the generalized moments \eqref{eq:generalized_moments} can be derived. Here, we present four different ways and illustrate all approaches in the one-dimensional case $\mathcal{X}=[\ul{x},\ol{x}]$. Similar formulas can be obtained for a multidimensional domain.\\

\noindent
\textbf{Probability density function}

For the derivation of $\EE[p_j(X_{t_{u+1}})\vert X_{t_u}=x_k]$, let the density function of the random variable $X_{t_{u+1}}\vert X_{t_u}=x_{k}$ be given as $f^{u,k}(x)$. Then, the conditional expectation can be derived by solving an integral,
\begin{align*}
\EE[p_j(X_{t_{u+1}})\vert X_{t_u}=x_k]
=\int_{\underline{x}}^{\overline{x}}T_{j}(\tau^{-1}_{[\ul{x},\ol{x}]}(y))\,f^{u,k}(y)\text{d}y
\end{align*} 
using $p_j(y)=T_j(\tau^{-1}_{\mathcal{X}}(y))1_{\mathcal{X}}(y)$. This approach is rather intuitive and easy to implement.\\ 

\noindent
\textbf{Fourier Transformation}

Assume the process $X$ has stationary increments and the characteristic function $\varphi$ of $X_{\Delta t}$ is explicitly available. We apply Parseval's identity, see \cite{Rudin1973}, and use Fourier transforms 
\begin{align*}
\EE[p_j(X_{t_{u+1}})\vert X_{t_u}=x_{k}]
=\int_{-\infty}^{\infty}p_{j}(x+x_{k})F(\dx)=\frac{1}{2\pi}\int_{-\infty}^{\infty}\widehat{p_{j}^{x_{k}}} (\xi)\varphi(-\xi)\text{d}\xi,
\end{align*}
where $p_{j}^{x_{k}}(x)=p_{j}(x+x_{k})$. Using the definition of $\tau_{[\underline{x},\overline{x}]}$, we can express the Fourier transform of $p_{j}^{x_{k}}(x)$ with the help of the Chebyshev polynomial $T_j(y)$. This yield
\begin{align}\label{eq:GM_with_Fourier}
\EE[p_j(X_{t_{u+1}})\vert X_{t_u}=x_{k}]
=\frac{1}{2\pi}e^{-i\xi x^k}e^{i\xi (\overline{x}-\frac{\overline{x}-\underline{x}}{2})}\frac{\overline{x}-\underline{x}}{2}\int_{-\infty}^{\infty}\widehat{T}_j\Big(\frac{\overline{x}-\underline{x}}{2}\xi\Big)\varphi(-\xi)\text{d}\xi.
\end{align}
The Fourier transform of the Chebyshev polynomials $\widehat{T_{j}}$ are presented in \cite{dominguez2011stability} and the authors also provide a Matlab implementation.\\ 

\noindent
\textbf{Truncated moments}

In this approach, we use that each one-dimensional Chebyshev polynomial can be represented as a sum of monomials, i.e.
\begin{align*}
T_{j}(x)=\sum_{l=0}^{j} a_{l} x^l,\ j\in\mathbb{N}.
\end{align*}
The coefficients $a_{l},\ l=0,\ldots,j$, can easily be derived using the \textit{chebfun} function \textit{poly()}, see \cite{driscoll2014chebfun}. Then, 
\begin{align*}
\EE[p_j(X_{t_{u+1}})\vert X_{t_u}=x_k]
&=\EE[T_j(\tau^{-1}_{\mathcal{X}}(X_{t_{u+1}}))1_{\mathcal{X}}(X_{t_{u+1}})\vert X_{t_u}=x_k]\\
&=\sum_{l=0}^{j} a_{l}\EE[(\tau^{-1}_{\mathcal{X}}(X_{t_{u+1}}))^{l}1_{\mathcal{X}}(X_{t_{u+1}})\vert X_{t_u}=x_k]
\end{align*} 
As $\tau_{\mathcal{X}}$ is linear the computation of the generalized moments has thus been reduced to deriving truncated moments.\\

\noindent
\textbf{Monte-Carlo simulation}

Lastly, especially in cases for which neither a probability density function, nor a characteristic function of the underlying process is given, Monte-Carlo simulation is a suitable choice. For every nodal point $x_{k}$ one simulates $N_{MC}$ paths $X_{t_{u+1}}^{i}$ of $X_{t_{u+1}}$ with starting value $X_{t_{u}}=x_{k}$. These simulations can then be used to approximate
\begin{align*}
\Gamma_{t_u,t_{u+1}}(p_j)(x^k)=\EE[p_j(X_{t_{u+1}})\vert X_{t_u}=x_{k}]
\approx\frac{1}{N_{MC}}\sum_{i=1}^{N_{MC}}p_j(X_{t_{u+1}}^{i})
\end{align*}
for every $j\in\mathcal{J}$. For an overview of Monte-Carlo simulation from SDEs and variance reduction techniques we refer to \cite{glasserman2003monte}.

\subsection{American Put Option}
In the numerical section we use the Dynamic Chebyshev method to price American put options. Assuming an asset model of the form $S_{t}=e^{X_{t}}$, the DPP becomes
\begin{align*}
V_{T}(x)&=(K-e^{x})^{+}\\
V_{t_u}(x)&=\max\left\{(K-e^{x})^{+}, e^{-r(t_{u+1}-t_{u})}\EE[V_{t_{u+1}}(X_{t_{u+1}})\vert X_{t_u}=x]\right\}.
\end{align*}
Typically, the support of the underlying process $X_{t}$ is $\mathbb{R}$ and restricting the domain to $\mathcal{X}$ introduces a truncation error. We reduce this error by exploiting the asymptotic behaviour of the payoff. If $X_{t_{u}}$ is below the exercise boundary the option is exercised at the value $K-e^{X_{t_{u}}}$ which we exploit for $X_{t_{u}}<\ul{x}$. The function $x\mapsto V_{t_{u}}$ tends to zero from above for $x\rightarrow\infty$ and thus for $\ol{x}$ large enough the truncation to zero for $x>\ol{x}$ is justified. Hence, we introduce the following modification of the Dynamic Chebyshev method:
\begin{align*}
V_{t_{u+1}}(x)&=V_{t_{u+1}}(x)1_{\{x<\ul{x}\}}+V_{t_{u+1}}(x)1_{\{x\in\mathcal{X}\}}+V_{t_{u+1}}(x)1_{\{x>\ol{x}\}}\\
&\approx (K-e^{x})1_{\{x<\ul{x}\}} + \widehat{V}_{t_{u+1}}(x)1_{\{x\in\mathcal{X}\}}
\end{align*}
and thus
\begin{align*}
\EE[V_{t_{u+1}}(X_{t_{u+1}})|X_{t_{u}}=x_{k}]&\approx \EE[(K-e^{X_{t_{u+1}}})1_{\{X_{t_{u+1}}<\ul{x}\}}|X_{t_{u}}=x_{k}]\\
&\quad + \sum_{j=0}^{N}c_{j}(t_{u+1})\Gamma_{j,k}(t_{u})
\end{align*}
for $\ul{x}$ small and $\ol{x}$ large enough. One can precompute $\EE[(K-e^{X_{t_{u+1}}})1_{\{X_{t_{u+1}}<\ul{x}\}}|X_{t_{u}}=x_{k}]$.  We emphasize that similar modifications to reduce the truncation error can be found for other payoff profiles, e.g. for digitals, butterfly options or any other combination of different put options.

Moreover, we also modify the first time step. Instead of approximating the payoff with Chebyshev polynomials at $t_{n}=T$ we just use the price of a European option to compute the continuation value at $t_{n-1}$. The kink of the payoff is in this case "smoothed" and convergence is improved.\\

The option's sensitivities Delta and Gamma can be computed by tanking the first or second derivative of
\begin{align*}
S\mapsto \widehat{V}_{0}(\log(S))=\sum_{j\in\mathcal{J}}c_{j}(t_0)p_{j}(\log(S)).
\end{align*}
Thus Delta and Gamma are expressed as the sum of derivatives of Chebyshev polynomials. In particular, their derivation comes without any additional computational costs.

\section{Numerical experiments}
In this section, we use the Dynamic Chebyshev method to price American put options and we numerically investigate the convergence of the method. Moreover, we compare the method with the Least-squares Monte-Carlo method of \cite{longstaffschwartz}.

\subsection{Stock price models}
For the convergence analysis we use three different stock price models.\\

\textbf{The Black-Scholes model:}\\
In the classical model of \cite{blackscholes} the stock price process is modelled by the SDE
\begin{align*}
\text{d}S_{t}=r S_{t}\text{d}t + \sigma S_{t}\text{d}W_{t}
\end{align*}
where $r$ is the risk-free interest rate and $\sigma>0$ is the volatility. 
In this model the log-returns $X_{t}=\log(S_t)$ are normally distributed and for the double truncated moments
\begin{align*}
\mathbb{E}[X^{m}1_{[a,b]}(X)] \qquad \text{for} \qquad X\sim\mathcal{N}(\mu,\sigma^{2})
\end{align*}
explicit formulas are available.\cite{KanRobotti2017} present results for the (multivariate) truncated moments and provide a Matlab implementation. \\

\textbf{The Merton jump diffusion model:}\\
The jump diffusion model introduced by \cite{merton1976} adds jumps to the classical Black-Scholes model. For $S_{t}=S_{0}e^{X_{t}}$ the log-returns $X_{t}$ follow a jump diffusion with volatility $\sigma$ and added jumps arriving at rate $\lambda>0$ with normal distributed jump sizes according to $\mathcal{N}(\alpha,\beta^{2})$. The characteristic function of $X_{t}$ is given by
\begin{align*}
\varphi(z)=exp\left(t\left(ibz - \frac{\sigma^{2}}{2}z^{2} + \lambda\left(e^{iz\alpha - \frac{\beta^{2}}{2}z^{2}}-1\right)\right)\right)
\end{align*} 
with risk-neutral drift $b=r-\frac{\sigma^{2}}{2}-\lambda\Big(e^{\alpha+\frac{\beta^{2}}{2}}-1\Big)$.\\

\textbf{The Constant Elasticity of Variance model:}\\
The Constant Elasticity of Variance model (CEV) as stated in \cite{Schroder1989} is a local volatility model based on the stochastic process
\begin{align}\label{CEV_model_SDE}
\text{d}S_{t}=r S_{t}\text{d}t + \sigma S_{t}^{\beta/2}\text{d}W_{t} \quad \text{for} \quad \beta>0.
\end{align}
Hence the stock volatility $\sigma S_{t}^{(\beta-2)/2}$ depends on the current level of the stock price. For the special case $\beta=2$ the model coincides with the Black-Scholes model. However, from market data one typically observes a $\beta<2$. The CEV-model is one example of a model which has neither a probability density, nor a characteristic function in closed-form.

\subsection{Convergence analysis}
In this section we investigate the convergence of the Dynamic Chebyshev method. We price American put options along with the options' Delta and Gamma in the Black-Scholes and the Merton jump diffusion model, where we can use the COS method of \cite{fang2009pricing} as benchmark. The COS method is based on the Fourier-cosine expansion of the density function and provides fast and accurate results for the class of L\'evy models.

For the experiments, we use the following parameter sets in the Black-Scholes model 
\begin{align*}
K=100, \quad r=0.03, \quad \sigma=0.25, \quad T=1,
\end{align*}
and for the Merton jump diffusion model
\begin{align*}
K=100, \quad r=0.03, \quad \alpha=-0.5, \quad \beta=0.4,\quad \sigma=0.25, \quad \lambda=0.4
\end{align*}
and we use $32$ time steps.\\

For both models the generalized moments are computed by the Fourier approach as stated in \eqref{eq:GM_with_Fourier}. We truncate the integral at $\vert \xi\vert \leq 250$ and use Clenshaw-Curtis with $500$ nodes for the numerical integration. For the Fourier transform of the Chebyshev polynomials the implementation of \cite{dominguez2011stability} is used. 
We run the Dynamic Chebyshev method for an increasing number of Chebyshev nodes $N=50, 100,\ldots,750$. Then, option prices and their sensitivities delta and gamma are calculated on a grid of different values of $S_{0}$ equally distributed between $60\%$ and $140\%$ of the strike $K$. The resulting prices and Greeks are compared using the COS method as benchmark and the maximum error over the grid is calculated. Here we use the implementation provided in \cite{von2015benchop}.\\

Figure \ref{fig:DC_error_decay_Four_Prices_Greeks} shows the error decay for the Black-Scholes model (left hand side) and the Merton model (right hand side). We observe that the method converges and an error below $10^{-3}$ is reached for $N=300$ Chebyshev nodes. The experiments confirm that the method can be used for an American put option.

\begin{figure}[H]
\centering
  \begin{minipage}{0.45\textwidth} 
     \centering 
     \includegraphics[width=\textwidth]{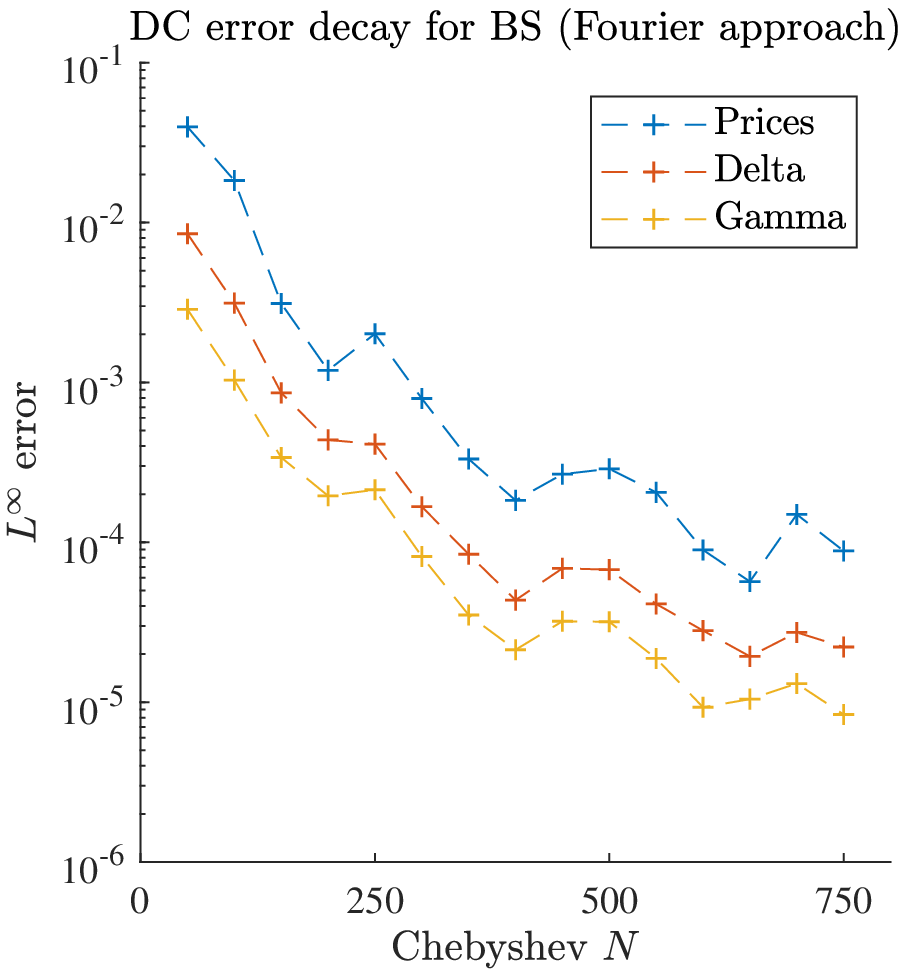} 
 \end{minipage} 
  \begin{minipage}{0.45\textwidth} 
     \centering 
     \includegraphics[width=\textwidth]{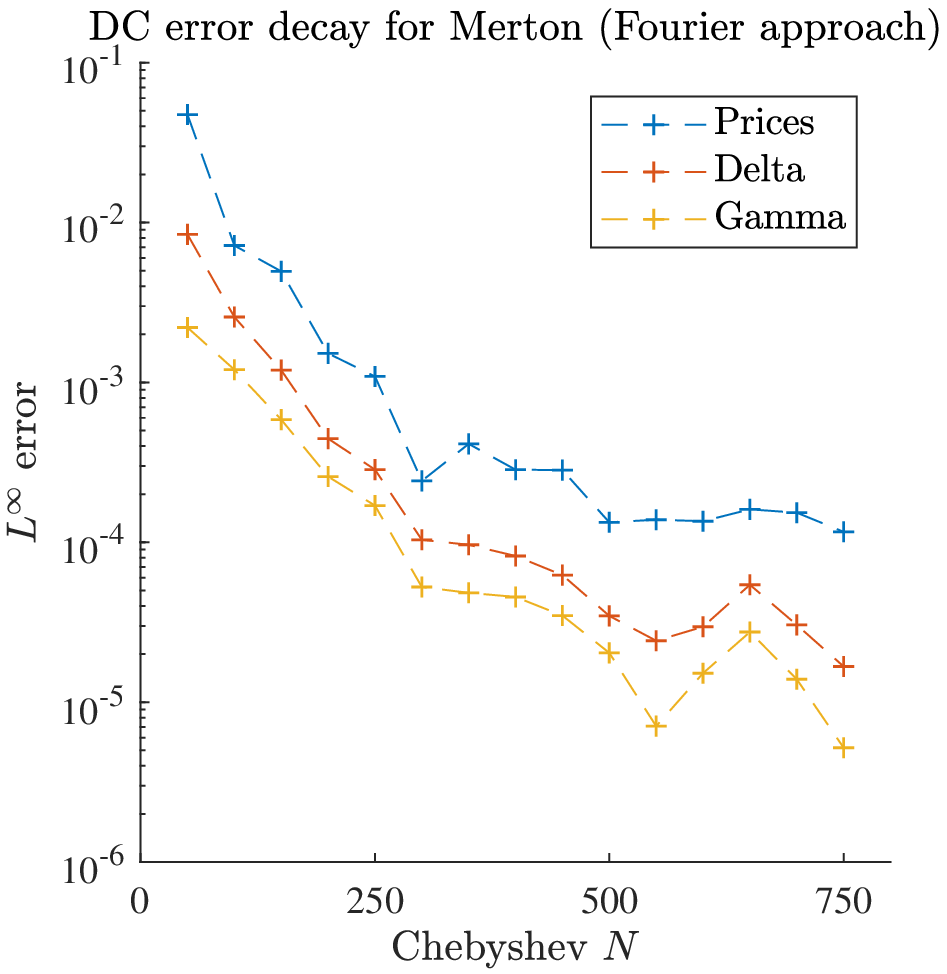} 
  \end{minipage} 
  \caption{Error decay prices Dynamic Chebyshev in the BS model (left) and the Merton model (right). The conditional expectation of the Chebyshev polynomials are calculated with the Fourier transformation.}
  \label{fig:DC_error_decay_Four_Prices_Greeks} 
\end{figure}


\FloatBarrier

\subsection{Dynamic Chebyshev with Monte-Carlo}
So far we empirically investigated the error decay of the method for option price and their derivatives. In this section we compare the Dynamic Chebyshev method with the Least Square Monte-Carlo approach of \cite{longstaffschwartz} in terms of accuracy and runtime. 

\subsubsection{The Black-Scholes model}
As a first benchmark, we use the Black-Scholes model with an interest rate of $r=0.03$ and volatility $\sigma=0.25$. Here, we look at a whole option price surface with varying maturities and strikes. We choose $12$ different maturities between one month and four years given by
\begin{align*}
T\in\{1/12,2/12,3/12,6/12,9/12,1,15/12,18/12,2,30/12,3,4\}
\end{align*} 
and strikes equally distributed between $80\%$ and $120\%$ of the current stock price $S_{0}=100$ in steps of $5\%$. We fix $n=504$ time steps (i.e. exercise rights) per year.\\

We compare the Dynamic Chebyshev method to the Least Squares Monte-Carlo approach. We run both methods for an increasing number of Monte-Carlo paths according to \begin{align}\label{number_MC_paths}
M\in\left\{2500,5000,10000,20000,40000,80000\right\}.
\end{align}
The convergence of the DC method depends on both, the number of nodes $N$ and the number of Monte-Carlo paths $M$. For an optimal convergence behaviour one needs to find a reasonable relationship between these factors. For the following experiments, we fix the number of Chebyshev nodes as $N=\sqrt{2}\sqrt{M}$.

Figure \ref{fig:BS_DC_surface_MC} shows the price surface and the error surface for $N=400$ and $M=80000$. The error was estimated by using the COS method as benchmark. We reach a maximal error below $4*10^{-2}$ on the whole option surface.\\

In Figure \ref{fig:BS_DC_LSM_error_runtime} the $log_{10}$-error is shown as a function of the $log_{10}$-runtime for both methods. The left figure compares the total runtimes and the right figure compares the offline runtime. For the Dynamic Chebyshev method the total runtime includes the offline-phase and the online phase. The offline-phase consists of the simulation of one time step of the underlying asset process $X_{\Delta t}$ for $N+1$ starting values $X_{0}=x_k$ and of the computation of the conditional expectations $E[p_{j}(X_{\Delta t})|X_{0}=x_k]$ for $j,k=0,\ldots,N$. The online phase is the actual pricing of the American option for all strikes and maturities. Similar, the total runtime of the Least-Square Monte-Carlo method includes the simulation of the Monte-Carlo paths (offline-phase) and the pricing of the option via backward induction (online-phase).\\

We observe that the Dynamic Chebyshev method reaches the same accuracy in a much lower runtime. For example, a maximum error of $0.1$ is reached in a total runtime of $0.5s$ with the Dynamic Chebyshev method whereas the LSM approach needs $98s$. This means the Dynamic Chebyshev method is nearly $200$ times faster for the same accuracy.
For the actual pricing in the online phase, the gain in efficiency is even higher. We observe that the Dynamic Chebyshev method outperforms the Least-Square Monte-Carlo method in terms of the total runtime and the pure online runtime. Moreover, we observe that the performance gain from splitting the computation into an offline and an online phase is much higher for the Dynamic Chebyshev method. For instance, in the example above the online runtime of the Dynamic Chebyshev method is $0.05s$ whereas the LSM takes $95s$, a factor of $1900$ times more.

\begin{figure} 
\centering
  \begin{minipage}{0.45\textwidth} 
     \centering 
     \includegraphics[width=\textwidth]{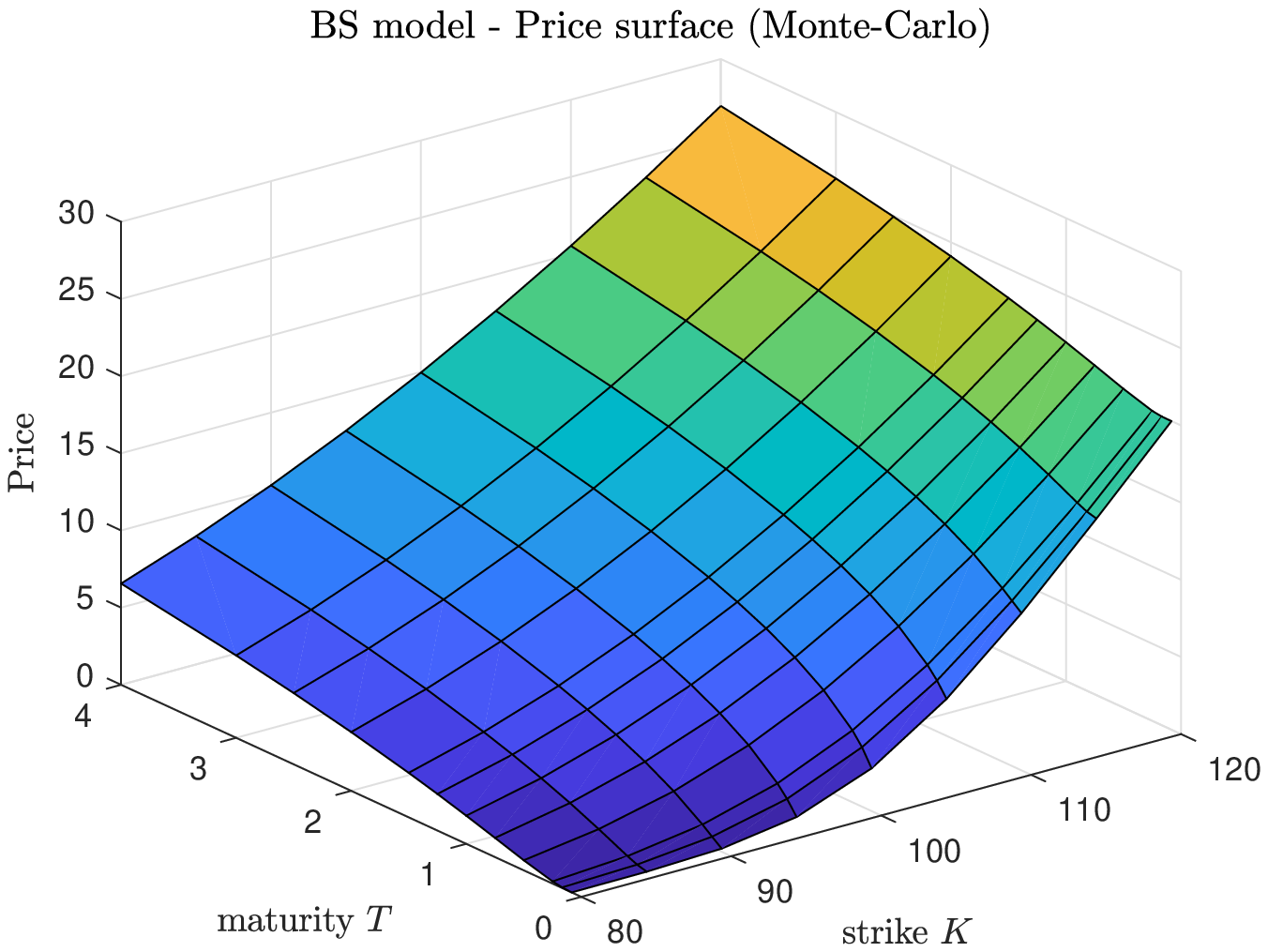} 
 \end{minipage} 
  \begin{minipage}{0.45\textwidth} 
     \centering 
     \includegraphics[width=\textwidth]{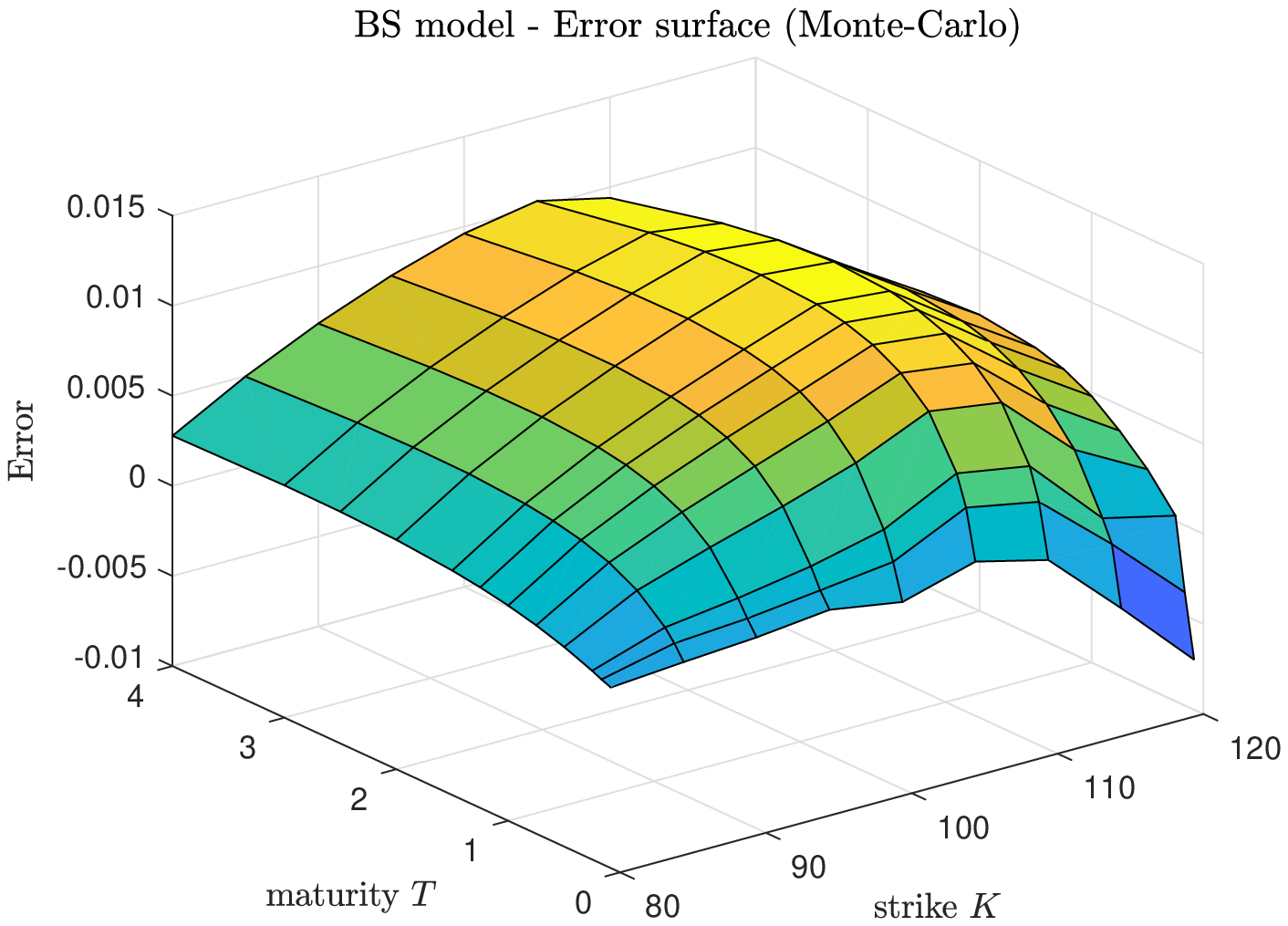} 
  \end{minipage} 
  \caption{Price surface and corresponding error of the Dynamic Chebyshev method in the Black-Scholes model. The conditional expectations are calculated with Monte-Carlo.
  }
  \label{fig:BS_DC_surface_MC} 
\end{figure}

\begin{figure} 
\centering
  \begin{minipage}{0.45\textwidth} 
     \centering 
     \includegraphics[width=\textwidth]{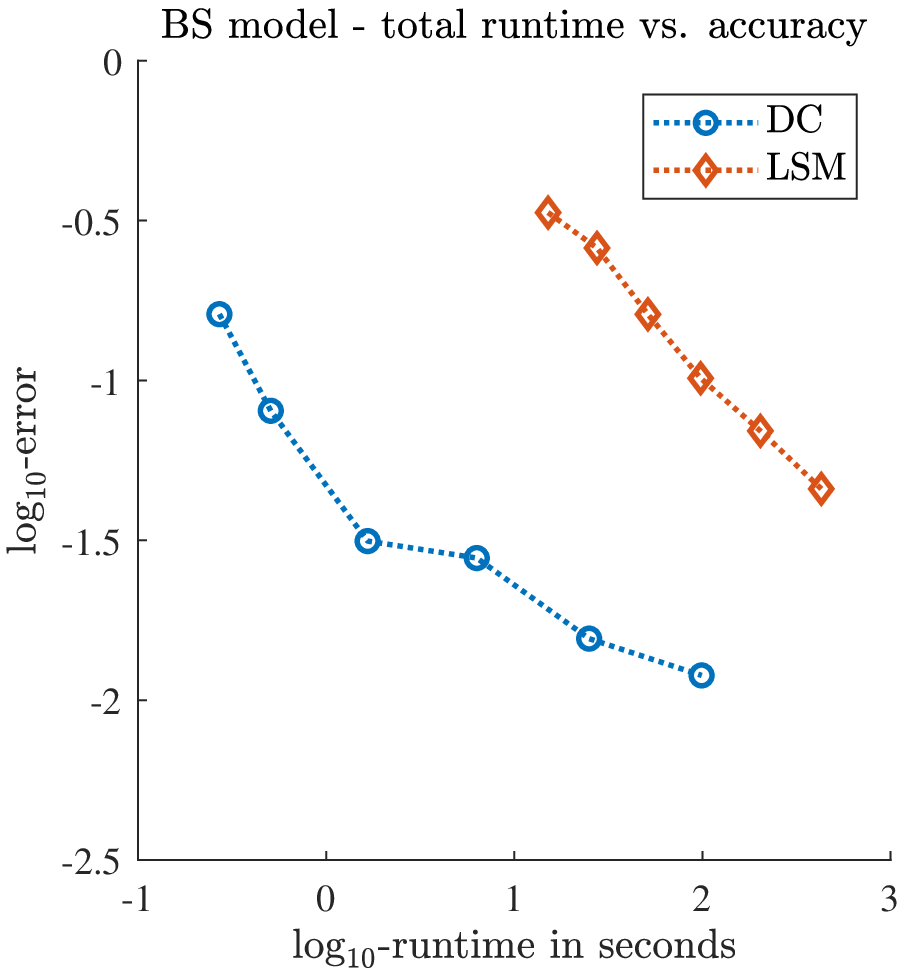} 
 \end{minipage} 
  \begin{minipage}{0.45\textwidth} 
     \centering 
     \includegraphics[width=\textwidth]{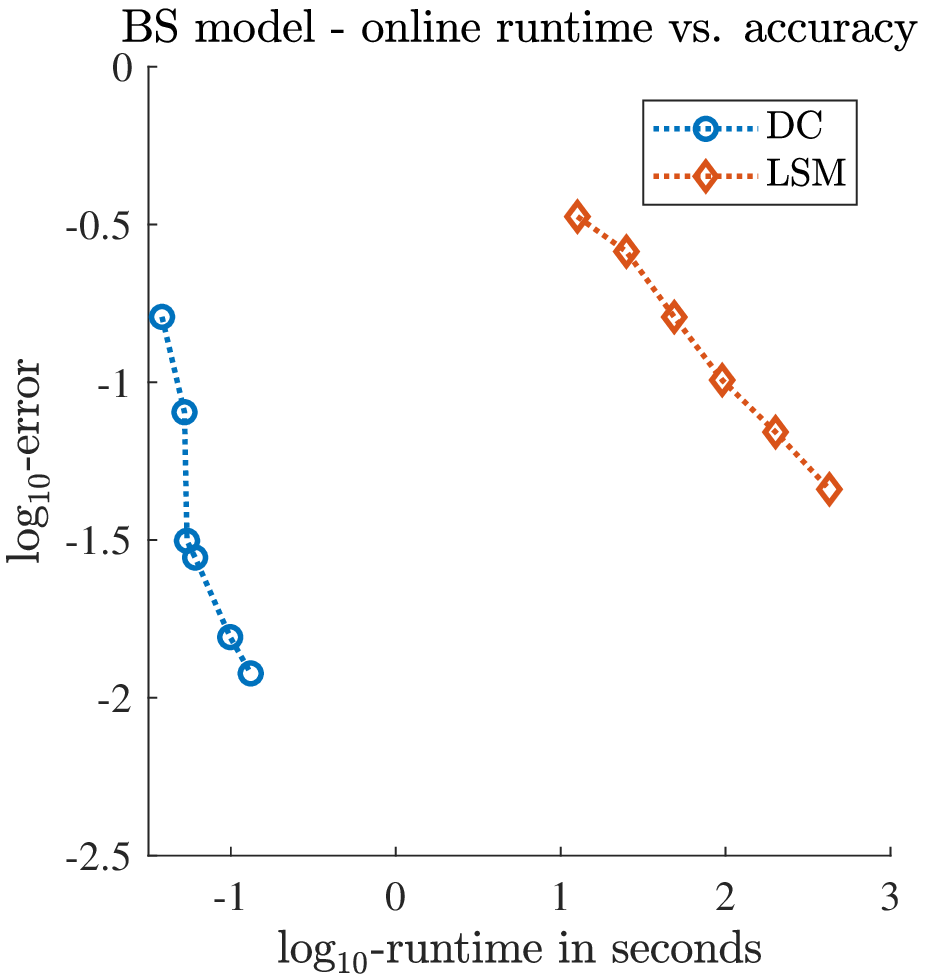} 
  \end{minipage} 
  \caption{Log-Log plot of the total/online runtime vs. accuracy. Comparison of the Dynamic Chebyshev method with the Least-Square Monte-Carlo algorithm.}
  \label{fig:BS_DC_LSM_error_runtime} 
\end{figure}


The main advantage of the Dynamic Chebyshev method is that once the conditional expectations are calculated, they can be used to price the whole option surface. The pure pricing, i.e. the online phase, is highly efficient. Furthermore, one only needs to simulate one time step $\Delta t$ of the underlying stochastic process instead of the complete path. We investigate this efficiency gain by varying the number of options and the number of time steps (exercise rights).
Figure \ref{fig:BS_DC_LSM_runtime_details} compares the total runtime of the DC method with the total runtime of the LSM method for an increasing number of options and for an increasing number time steps. As expected, we can empirically confirm that the efficiency gain by the Dynamic Chebyshev methods increases with number of options and the number of exercise rights. In both cases, the runtime of the DC method stays nearly constant whereas the runtime of the LSM method increases approximately linearly.
\begin{figure}[H] 
\centering
  \begin{minipage}{0.45\textwidth} 
     \centering 
     \includegraphics[width=\textwidth]{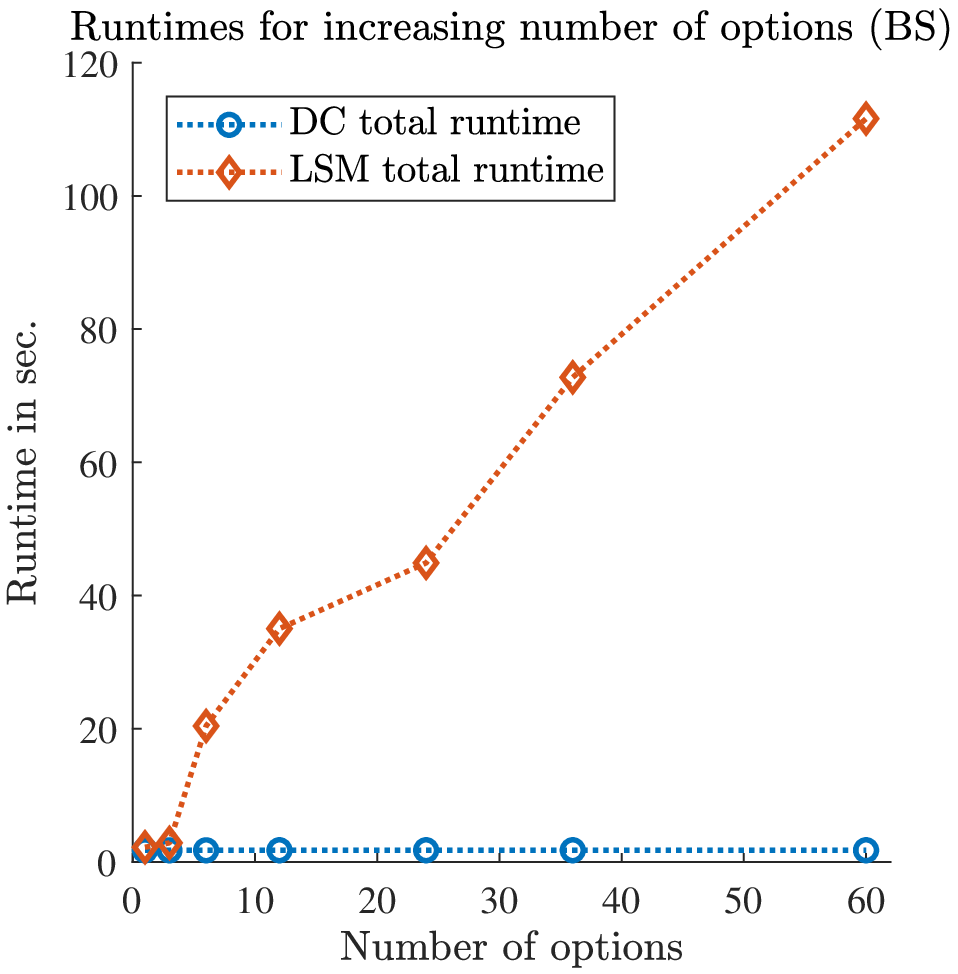} 
 \end{minipage} 
  \begin{minipage}{0.45\textwidth} 
     \centering 
     \includegraphics[width=\textwidth]{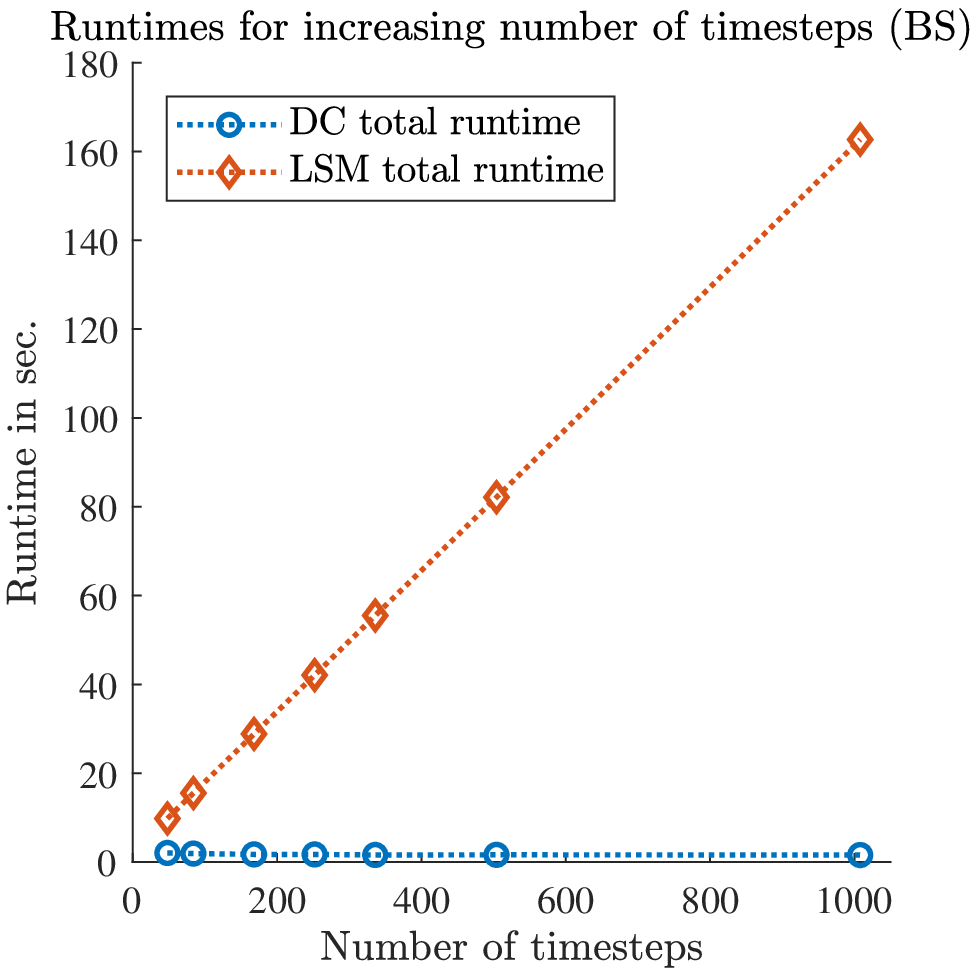} 
  \end{minipage} 
  \caption{Total runtime of the DC and the LSM method for an increasing number of options (left) and an increasing number of timesteps (right).}
  \label{fig:BS_DC_LSM_runtime_details} 
\end{figure}
\FloatBarrier



\subsubsection{The CEV model}
Next, we use the constant elasticity of variance (CEV) model for the underlying stock price process. We perform the same experiments as in the last section. The parameters in the CEV model, as in \eqref{CEV_model_SDE} are the following
\begin{align*}
\sigma=0.25,\quad r=0.03,\quad \beta=1.5.
\end{align*}


Similarly, we compare the Dynamic Chebyshev and the LSM method by computing the prices of an option price surface. We use the same parameter specifications for $K$, $T$ and $n$. We run both methods for an increasing number of Monte-Carlo simulations $M$ and fix $N=\sqrt{2}\sqrt{M}$.
Figure \ref{fig:CEV_DC_surface_MC}  shows the price surface and the error surface for $N=400$ and $M=80000$. The error is calculated using a binomial tree implementation of the CEV model based on \cite{NelsonRamaswamy1990}.\\

\begin{figure}[H]
\centering
  \begin{minipage}{0.45\textwidth} 
     \centering 
     \includegraphics[width=\textwidth]{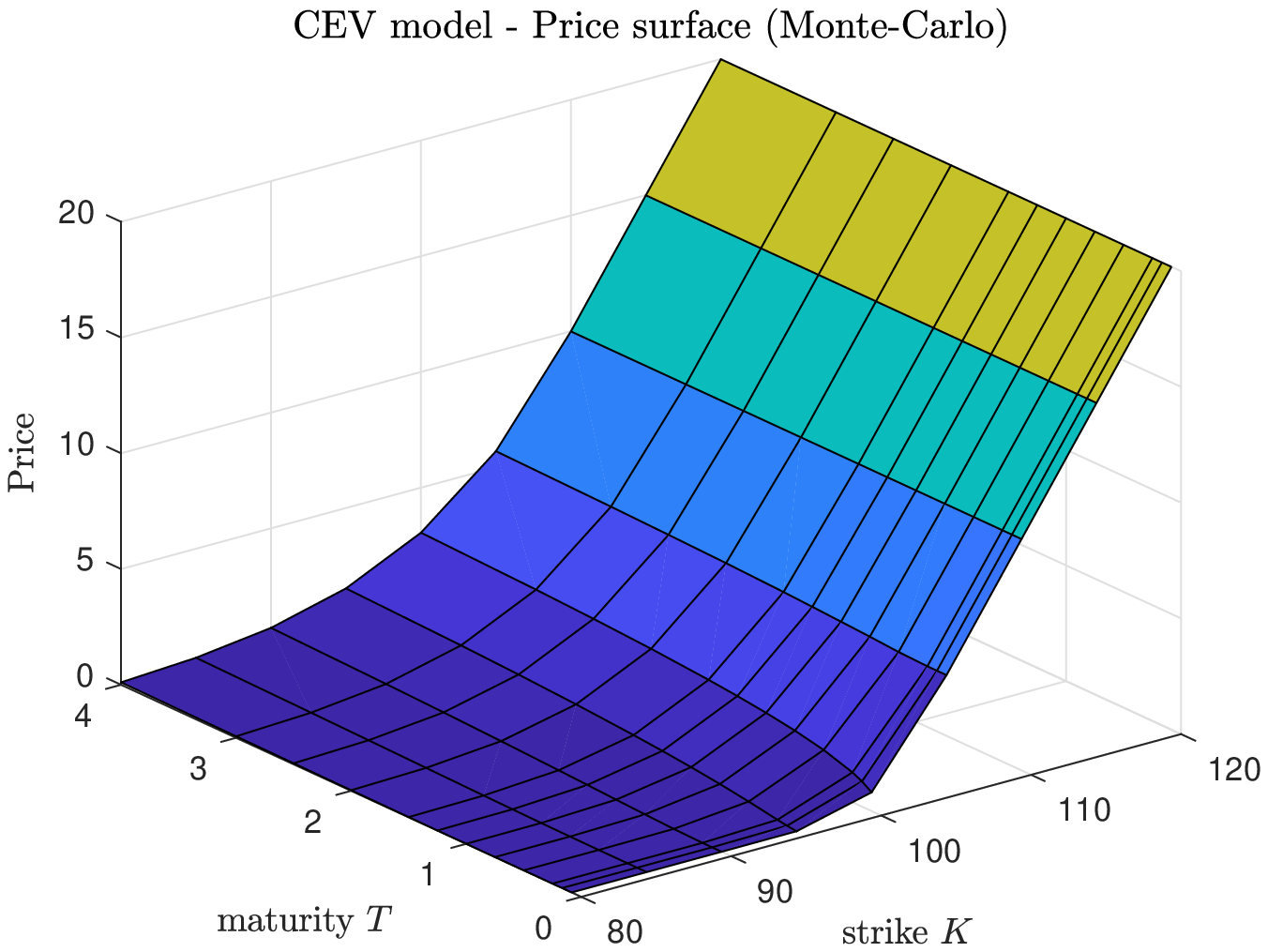} 
 \end{minipage} 
  \begin{minipage}{0.45\textwidth} 
     \centering 
     \includegraphics[width=\textwidth]{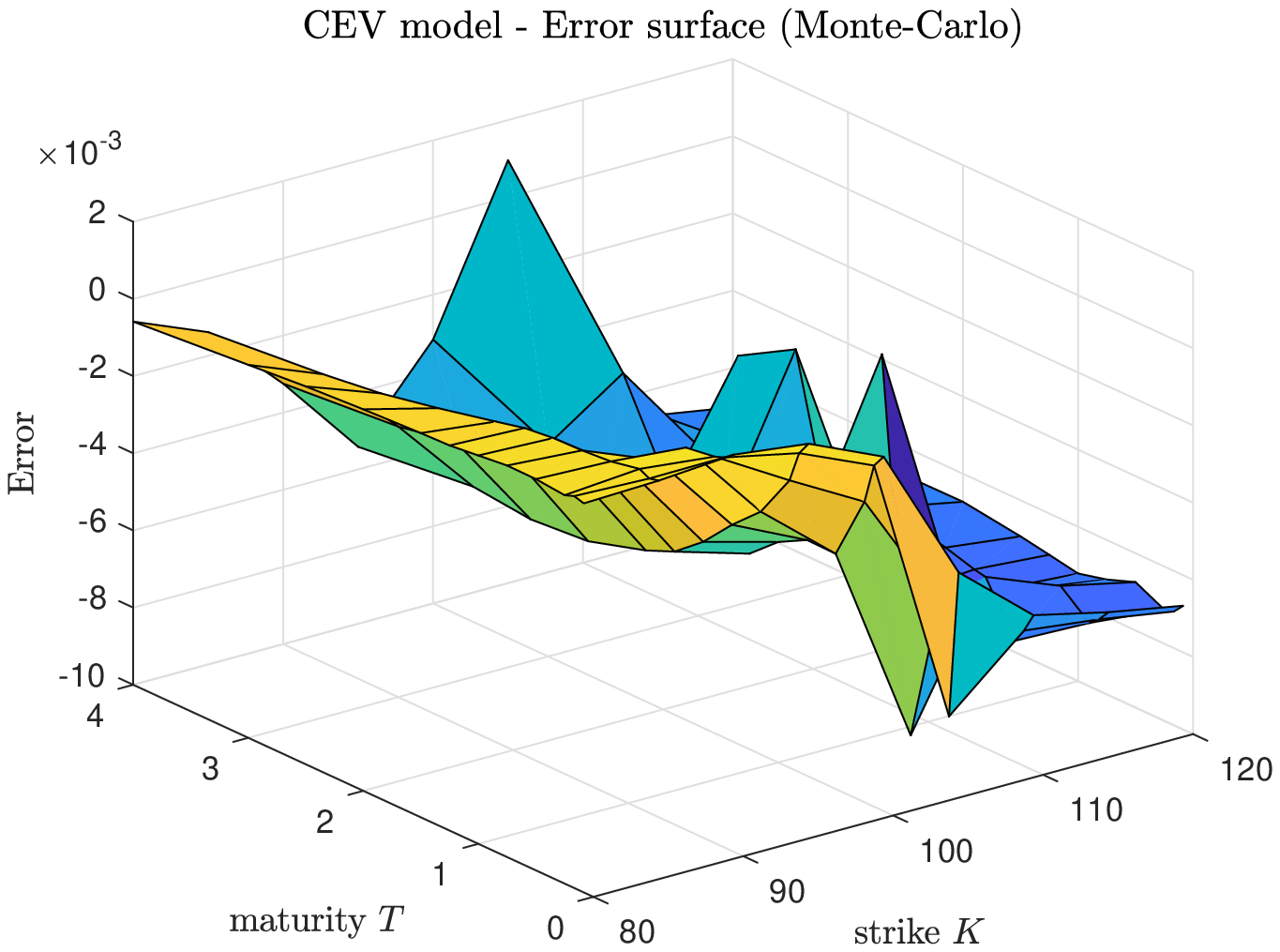} 
  \end{minipage} 
  \caption{Price surface and corresponding error of the Dynamic Chebyshev method in the CEV model. The conditional expectations are calculated with Monte-Carlo.}
  \label{fig:CEV_DC_surface_MC} 
\end{figure} 

In Figure \ref{fig:CEV_DC_LSM_error_runtime} the $log_{10}$-error is shown as a function of the $log_{10}$-runtime for both methods. The left figure compares the total runtimes and the right figure compares the offline runtimes. Again, we observe that the Dynamic Chebyshev method is faster for the same accuracy and it profits more from an offline-online decomposition. For example, the total runtime of the DC method to reach an accuracy below $0.03$ is $3.5s$ whereas LSM takes $136s$. For the online runtimes this out-performance is $1s$ to $122s$.
\begin{figure}[H] 
\centering
  \begin{minipage}{0.45\textwidth} 
     \centering 
     \includegraphics[width=\textwidth]{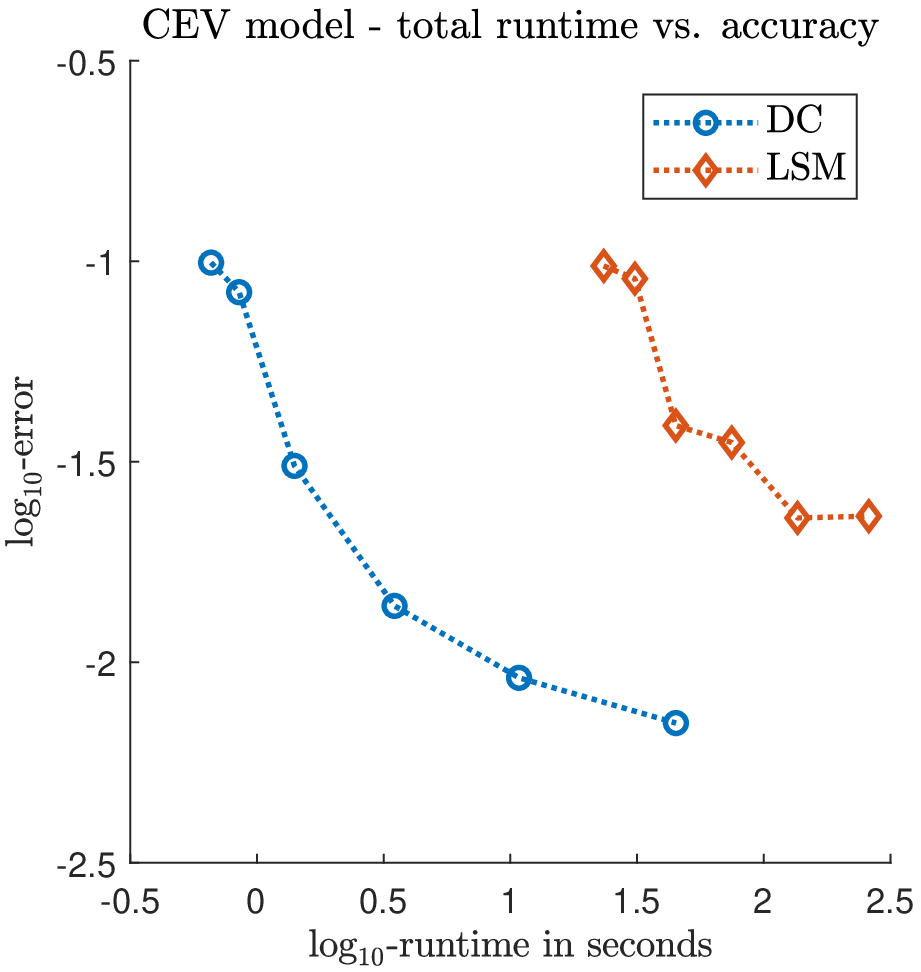} 
 \end{minipage} 
  \begin{minipage}{0.45\textwidth} 
     \centering 
     \includegraphics[width=\textwidth]{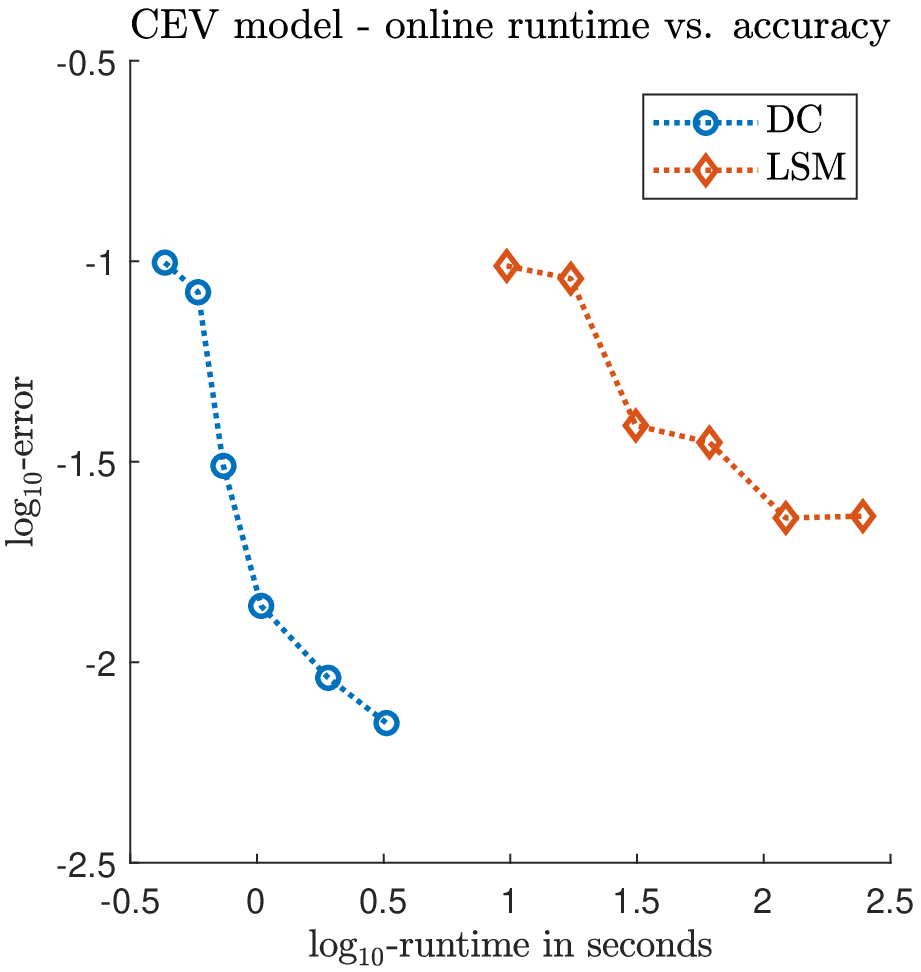} 
  \end{minipage} 
  \caption{Log-Log plot of the total/online runtime vs. accuracy. Comparison of the Dynamic Chebyshev method with the Least-Square Monte-Carlo algorithm.}
  \label{fig:CEV_DC_LSM_error_runtime} 
\end{figure}

Investigating this efficiency gain further, we look at the performance for different numbers of options and time steps (exercise rights). Similarly to the last section, Figure \ref{fig:CEV_DC_LSM_runtime_details} compares the total runtime of the DC method with the total runtime of the LSM method for an increasing number of options and time steps. In both cases, the runtime of the DC method stays nearly constant whereas the runtime of the LSM method increases approximately linearly. This observation is consistent with the findings for the Black-Scholes model in the previous section.
\begin{figure}[H] 
\centering
  \begin{minipage}{0.45\textwidth} 
     \centering 
     \includegraphics[width=\textwidth]{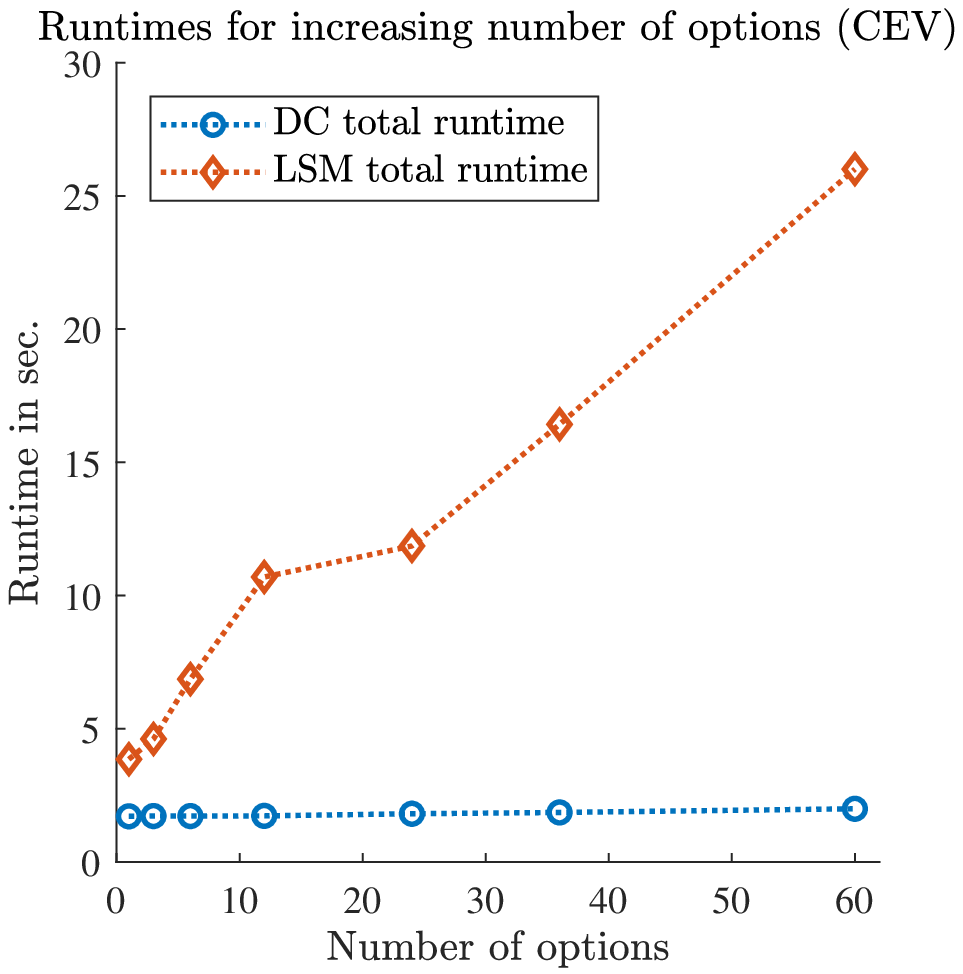} 
 \end{minipage} 
  \begin{minipage}{0.45\textwidth} 
     \centering 
     \includegraphics[width=\textwidth]{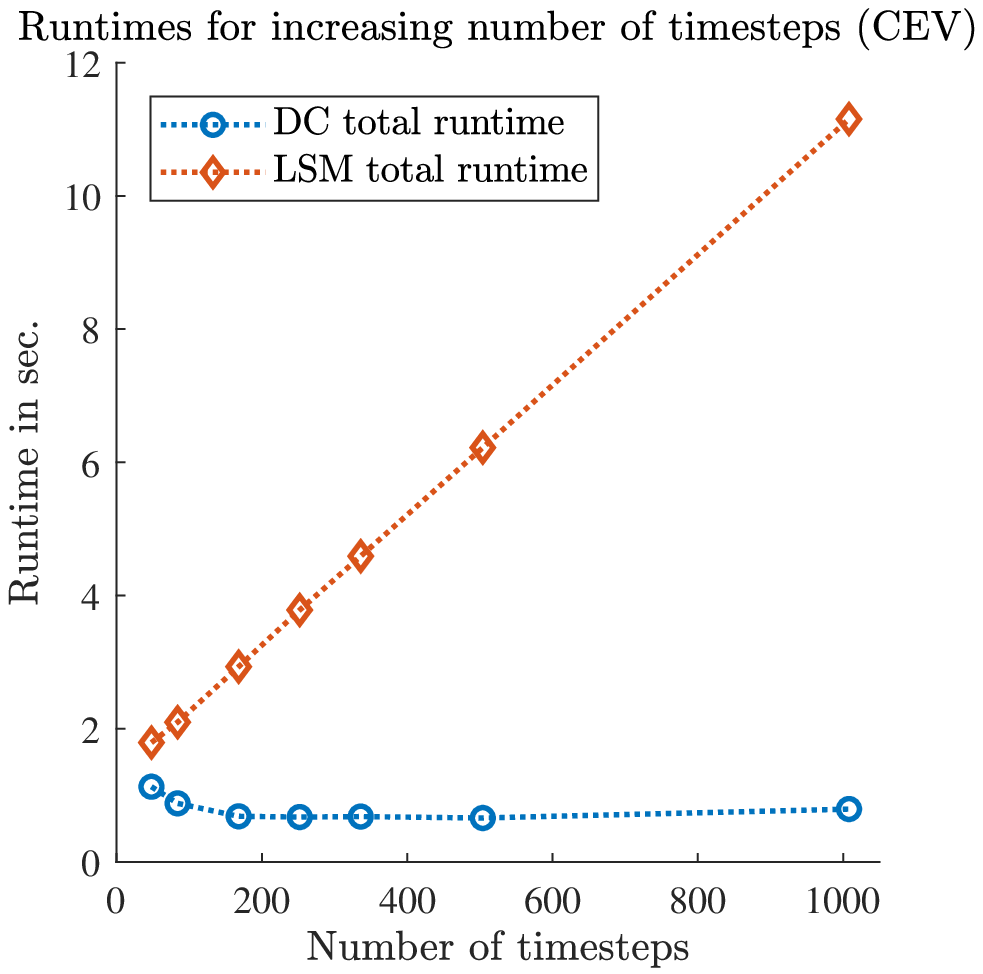} 
  \end{minipage} 
  \caption{Total runtime of the DC and the LSM method for an increasing number of options (left) and an increasing number of timesteps (right).}
  \label{fig:CEV_DC_LSM_runtime_details} 
\end{figure}

%

\FloatBarrier
\section{Conclusion and Outlook}
We have introduced a new approach to price American options via backward induction by approximating the value function with Chebyshev polynomials. Thereby, the computation of the conditional expectation of the value function in each time step is reduced to the computation of conditional expectations of polynomials. The proposed method separates the pricing of an option into a part which is model-dependent (the computation of the conditional expectations) and the pure pricing of a given payoff which becomes independent of the underlying model. The first step, the computation of the conditional expectation of the Chebyshev polynomials, is the so-called offline phase of the method. The design of the method admits several qualitative advantageous: 
\begin{itemize}
\item If the conditional expectations are set-up once, we can use them for the pricing of many different options. Thus, the actual pricing in the online step becomes very simple and fast. 
\item In the pre-computation step one can combine the method with different techniques, such as Fourier approaches and Monte-Carlo simulation. Hence the method can be applied in a variety of models.
\item The proposed approach is very general and flexible and thus not restricted to the pricing of American options. It can be used to solve a large class of optimal stopping problems.
\item We obtain a closed-form approximation of the option price as a function of the stock price at every time step. This approximation can be used to compute the option's sensitivities Delta and Gamma at no additional costs. This holds for all models and payoff profiles, even if Monte-Carlo is required in the offline phase.
\item The method is easy to implement and to maintain. The pre-computation step is well-suited for parallelization to speed-up the method.
\end{itemize}
We have investigated the theoretical error behaviour of the method and introduced explicit error bounds. We put particular emphasis on the combination of the method with Monte-Carlo simulation. Numerical experiments confirm that the method performs well for the pricing of American options. A detailed comparison of the method with the Least-Square Monte-Carlo approach proposed by \cite{longstaffschwartz} confirmed a high efficiency gain. Especially, when a high number of options is priced, for example a whole option price surface. In this case, the Dynamic Chebyshev method highly profits from the offline-online decomposition. Once the conditional expectations are computed they can be used to price options with different maturities and strikes. Besides the efficiency gain, the closed-form approximation of the price function is a significant advantage because it allows us to calculate the sensitivities. Since \cite{longstaffschwartz} introduced their method different modifications have been introduced. Either to increase efficiency or to tackle the sensitivities. For example the simulation algorithm of \cite{JainOosterlee2015} is comparable to LSM in terms of efficiency but is able to compute the Greeks at no additional costs. Moreover dual approaches were developed to obtain upper bounds for the option price, see \cite{Rogers2002} and more recently \cite{BelomestnyHaefnerUrusov2018}.\\

The presented error analysis of the method under an analyticity assumption is the starting point for further theoretical investigations in the case of piecewise analyticity and (piecewise) differentiability. The former allows to cover rigorously the American option pricing problem and a preliminary version is presented in \cite{Mahlstedt2017}. The qualitative merits of the method can be exploited in a variety of applications. \cite{GlauPachonPoetz2018} take advantage of the closed-form approximation to efficiently compute the expected exposure of early-exercise options as a step in CVA calculation. Moreover, the method can be used to price different options such as Barrier options, Swing options or multivariate American options.

\appendix 

\section{Proof of Theorem 3.3}
\begin{proof}
Consider a DPP as defined in Definition \ref{defin_DPP}, i.e. we have a Lipschitz continuous function
\begin{align*}
\vert f(x_{1},y_{1})-f(x_{2},y_{2}) \vert \leq L_{f}(\vert x_{1}-x_{2} \vert + \vert y_{1} - y_{2} \vert).
\end{align*}
Assume that the regularity Assumption \ref{assumption_analytic_value} and the assumption on the truncation error \eqref{trunc_err_bound} hold. Then we have to distinguish between the deterministic case \eqref{cond_exp_delta_bound} and the stochastic case \eqref{cond_exp_MC_bound}. In the first case, the expectation in the error bound can simply be ignored. First, we apply Proposition \ref{Cheby_Err_Distortion_multi}. At time point $T$ there is no random part and no distortion error. Thus,
\begin{align*}
\max_{x\in\mathcal{X}}\EE\left[|V_{T}(x)-\widehat{V}_{T}(x)|\right]=\max_{x\in\mathcal{X}}|V_{T}(x)-\widehat{V}_{T}(x)|\leq \varepsilon_{int}(\varrho_{t_{n}},N,D,M_{t_{n}}).
\end{align*}
For the ease of notation we will from know on write $\varepsilon^{j}_{int}=\varepsilon_{int}(\varrho_{t_{j}},N,D,M_{t_{j}})$. We obtain for the error at $t_{u}$
\begin{align}\label{eq:proof_error_bound_distortion}
\max_{x\in\mathcal{X}}\EE\left[|V_{t_u}(x)-\widehat{V}_{t_u}(x)|\right]
\leq \varepsilon^{u}_{int} + \Lambda_{\ol{N}}F(f,t_{u})
\end{align}
with maximal distortion error $F(f,t_{u})=\max_{k\in\mathcal{J}}\EE\left[|V_{t_{u}}(x^{k})-\widehat{V}_{t_{u}}(x^{k})|\right]$. Note that whether \eqref{cond_exp_delta_bound} or \eqref{cond_exp_MC_bound} hold an approximation error of the conditional expectation of $\widehat{V}_{t_{u+1}}$ is made, i.e.
$\EE[\widehat{V}_{t_{u+1}}(X_{t_{u+1}})\vert X_{t_u}=x^k]=\Gamma_{t_{u}}^{k}(\widehat{V}_{t_{u+1}})\approx\widehat{\Gamma}_{t_{u}}^{k}(\widehat{V}_{t_{u+1}})$. The Lipschitz continuity of $f$ yields
\begin{align*}
\left\vert V_{t_{u}}(x^k)- \widehat{V}_{t_{u}}(x^k)\right\vert
&=\Big\vert f\left(g(t_u,x^k), \Gamma_{t_{u}}^{k}(V_{t_{u+1}})\right)
\ - f\Big(g(t_u,x^k),\widehat{\Gamma}_{t_{u}}^{k}(\widehat{V}_{t_{u+1}})\Big)\Big\vert\\
&\leq L_{f}\Big(\Big\vert g(t_u,x^k)-g(t_u,x^k)\Big\vert 
\ +\Big\vert \Gamma_{t_{u}}^{k}(V_{t_{u+1}})-\ \widehat{\Gamma}_{t_{u}}^{k}(\widehat{V}_{t_{u+1}})\Big\vert\Big)\\
&= L_{f}\Big(\Big\vert \Gamma_{t_{u}}^{k}(V_{t_{u+1}})-\ \widehat{\Gamma}_{t_{u}}^{k}(\widehat{V}_{t_{u+1}})\Big\vert\Big)\\
&\leq L_{f}\Big(\Big\vert \Gamma_{t_{u}}^{k}(V_{t_{u+1}}1_{\mathcal{X}})-\Gamma_{t_{u}}^{k}(\widehat{V}_{t_{u+1}})\Big\vert + \Big\vert \Gamma_{t_{u}}^{k}(V_{t_{u+1}}1_{\mathbb{R}^{D}\setminus\mathcal{X}})\Big\vert\\
&\qquad + \Big\vert \Gamma_{t_{u}}^{k}(\widehat{V}_{t_{u+1}})-\widehat{\Gamma}_{t_{u}}^{k}(\widehat{V}_{t_{u+1}})\Big\vert\Big).
\end{align*}
Next, we consider the expectation for each of the three error terms. For the first term we obtain
\begin{align*}
\EE\left[\Big\vert \Gamma_{t_{u}}^{k}(V_{t_{u+1}}1_{\mathcal{X}})-\Gamma_{t_{u}}^{k}(\widehat{V}_{t_{u+1}})\Big\vert\right]&=\EE\left[\Big\vert \EE[V_{t_{u+1}}(X_{t_{u+1}})\1_{\mathcal{X}}-\widehat{V}_{t_{u+1}}(X_{t_{u+1}})\vert X_{t_u}=x^k]\Big\vert\right]\\
&\leq\max_{x\in\mathcal{X}}\EE\left[|V_{t_{u+1}}(x)-\widehat{V}_{t_{u+1}}(x)|\right]
=\varepsilon_{t_{u+1}}
\end{align*}
and for the second term we have
\begin{align*}
\EE\left[\Big\vert \Gamma_{t_{u}}^{k}(V_{t_{u+1}}1_{\mathbb{R}^{D}\setminus\mathcal{X}})\Big\vert\right]\leq\EE\left[\varepsilon_{tr}\right]= \varepsilon_{tr}.
\end{align*}
For the last term we have to distinguish two cases. If we assume \eqref{cond_exp_delta_bound} holds,
the operator norm yields
\begin{align*}
\Big\vert \Gamma_{t_{u}}^{k}(\widehat{V}_{t_{u+1}})-\widehat{\Gamma}_{t_{u}}^{k}(\widehat{V}_{t_{u+1}})\Big\vert&=\Big\vert \left(\Gamma_{t_{u}}^{k} - \widehat{\Gamma}_{t_{u}}^{k}\right)\Big(\widehat{V}_{t_{u+1}}\Big)\Big\vert\\
&\leq \Big\vert\Big\vert\Gamma_{t_{u}}^{k} - \widehat{\Gamma}_{t_{u}}^{k}\Big\vert\Big\vert_{op}\ \Big\vert\Big\vert\widehat{V}_{t_{u+1}}\Big\vert\Big\vert_{\infty}\\
&\leq \ol{\delta}\ \Big\vert\Big\vert\widehat{V}_{t_{u+1}}\Big\vert\Big\vert_{\infty}.
\end{align*}
Next, we consider the second case and assume that \eqref{cond_exp_MC_bound} holds. Then we have
\begin{align*}
\EE\left[\Big\vert \Gamma_{t_{u}}^{k}(\widehat{V}_{t_{u+1}})-\widehat{\Gamma}_{t_{u}}^{k}(\widehat{V}_{t_{u+1}})\Big\vert\right]\leq \Big\vert\Big\vert\Gamma_{t_{u}}^{k} - \widehat{\Gamma}_{t_{u}}^{k}\Big\vert\Big\vert_{op}\ \Big\vert\Big\vert\widehat{V}_{t_{u+1}}\Big\vert\Big\vert_{\infty}^{\star}
&\leq \delta^{\star}(M)\ \Big\vert\Big\vert\widehat{V}_{t_{u+1}}\Big\vert\Big\vert_{\infty}^{\star}.
\end{align*}
Hence in either case the following bound holds
\begin{align*}
\EE\left[\Big\vert \Gamma_{t_{u}}^{k}(\widehat{V}_{t_{u+1}})-\widehat{\Gamma}_{t_{u}}^{k}(\widehat{V}_{t_{u+1}})\Big\vert\right]\leq \varepsilon_{gm}\max_{x\in\mathcal{X}}\EE\left[|\widehat{V}_{t_{u+1}}(x)|\right]
\end{align*}
with $\varepsilon_{gm}=\ol{\delta}$ if assumption \eqref{cond_exp_delta_bound} holds and $\varepsilon_{gm}=\delta^{\star}(M)$ if assumption \eqref{cond_exp_MC_bound} holds. We need an upper bound for the maximum of the Chebyshev approximation
\begin{align*}
\max_{x\in\mathcal{X}}\EE\left[|\widehat{V}_{t_{u+1}}(x)|\right]&\leq \max_{x\in\mathcal{X}} \EE\left[|\widehat{V}_{t_{u+1}}(x)-V_{t_{u+1}}(x)|\right]+\max_{x\in\mathcal{X}}|V_{t_{u+1}}(x)|\leq \varepsilon_{t_{u+1}} + \ol{V}_{u+1}
\end{align*}
with $\ol{V}_{u+1}:=\max_{x\in\mathcal{X}}|V_{t_{u+1}}(x)|$. 
Hence, the error bound \eqref{eq:proof_error_bound_distortion} becomes
\begin{align*}
\varepsilon_{t_{u}}\leq\varepsilon^{u}_{int} +\Lambda_{\ol{N}} L_{f}\big((1+\varepsilon_{gm})\varepsilon_{t_{u+1}}+\varepsilon_{tr}+\varepsilon_{gm}\ol{V}_{u+1}\big).
\end{align*}
By induction, we now show \eqref{Error_t_k_Lip_delta}. For $u=n$ we have
$\varepsilon_{t_{n}}\leq\varepsilon^{n}_{int}$ and therefore \eqref{Error_t_k_Lip_delta} holds for u=n. We assume that for $n,\ldots,u+1$ equation \eqref{Error_t_k_Lip_delta} holds. For the error $\varepsilon_{t_{u}}$ we obtain
\begin{align*}
\varepsilon_{t_{u}}&\leq\varepsilon^{u}_{int} +\Lambda_{\ol{N}} L_{f}\big((1+\varepsilon_{gm})\varepsilon_{t_{u+1}}+\varepsilon_{tr}+\varepsilon_{gm}\ol{V}_{u+1}\big)\\
&\leq \varepsilon^{u}_{int} +\Lambda_{\ol{N}} L_{f}\Big((1+\varepsilon_{gm})\Big(\sum_{j=u+1}^{n}C^{j-(u+1)}\varepsilon^{j}_{int}\ +\ \Lambda_{\ol{N}}L_{f}\sum_{j=u+2}^{n}C^{j-(u+2)}(\varepsilon_{tr}+\varepsilon_{gm}\ol{V}_{j})\Big)\\
&\quad +\varepsilon_{tr}+\varepsilon_{gm}\ol{V}_{u+1}\Big)\\
&=\varepsilon^{u}_{int} + C\sum_{j=u+1}^{n}C^{j-(u+1)}\varepsilon^{j}_{int}\ + \Lambda_{\ol{N}} L_{f}\Big( C\sum_{j=u+2}^{n}C^{j-(u+2)}(\varepsilon_{tr}+\varepsilon_{gm}\ol{V}_{j})\\
&\quad + \varepsilon_{tr}+\varepsilon_{gm}\ol{V}_{u+1}\Big)\\
&=\varepsilon^{u}_{int} + \sum_{j=u+1}^{n}C^{j-u}\varepsilon^{j}_{int}\ + \Lambda_{\ol{N}} L_{f}\Big(\sum_{j=u+2}^{n}C^{j-(u+1)}(\varepsilon_{tr}+\varepsilon_{gm}\ol{V}_{j}) + \varepsilon_{tr}+\varepsilon_{gm}\ol{V}_{u+1}\Big)\\
&=\sum_{j=u}^{n}C^{j-u}\varepsilon^{j}_{int}\ + \Lambda_{\ol{N}} L_{f}\sum_{j=u+1}^{n}C^{j-(u+1)}(\varepsilon_{tr}+\varepsilon_{gm}\ol{V}_{j})
\end{align*}
which was our claim.
\end{proof}

\bibliographystyle{chicago}
  \bibliography{MMDiss_Literature}

\end{document}